\documentclass[sigconf]{acmart}

\AtBeginDocument{%
  }

\setcopyright{acmcopyright}
\copyrightyear{2024}
\acmYear{2024}
\acmDOI{XXXXXXX.XXXXXXX}

\acmConference[Conference acronym 'XX]{Make sure to enter the correct
  conference title from your rights confirmation emai}{June 03--05,
  2018}{Woodstock, NY}
\acmPrice{15.00}
\acmISBN{978-1-4503-XXXX-X/18/06}


\usepackage{color}
\usepackage{enumerate}
\usepackage{enumitem}
\usepackage[ruled,noend,linesnumbered]{algorithm2e}
\SetKw{Continue}{continue}
\SetKw{Break}{break}
\usepackage{url}%
\usepackage{multirow}
\usepackage{pbox}
\usepackage{bbding}
\usepackage{makecell}
\usepackage{soul}

\newlength{\figwidths}
\newlength{\expwidths}
\newlength{\expwidthd}
\newtheorem{problem}{Problem}

\usepackage{todonotes}
\setuptodonotes{inline}

\newenvironment{revised}{\color{black}}{}

\DeclareMathOperator*{\argmax}{arg\,max}

\begin{document}

\title{Multivariate Time Series Cleaning under Speed Constraints}

\author{Aoqian Zhang}
\affiliation{%
  \institution{Beijing Institute of Technology}
  \country{China}}
\email{aoqian.zhang@bit.edu.cn}

\author{Zexue Wu}
\affiliation{%
  \institution{Beijing Institute of Technology}
  \country{China}}
\email{zexue.wu@bit.edu.cn}

\author{Yifeng Gong}
\affiliation{%
  \institution{Beijing Institute of Technology}
  \country{China}}
\email{yifeng.gong@bit.edu.cn}

\author{Ye Yuan}
\affiliation{%
  \institution{Beijing Institute of Technology}
  \country{China}}
\email{yuan-ye@bit.edu.cn}

\author{Guoren Wang}
\affiliation{%
  \institution{Beijing Institute of Technology}
  \country{China}}
\email{wanggr@bit.edu.cn}

\renewcommand{\shortauthors}{Zhang et al.}

\begin{abstract}
Errors are common in time series due to unreliable sensor measurements.
Existing methods focus on univariate data but do not utilize the correlation between dimensions.
Cleaning each dimension separately may lead to a less accurate result, as some errors can only be identified in the multivariate case.
We also point out that the widely used minimum change principle is not always the best choice.
Instead, we try to change the smallest number of data to avoid a significant change in the data distribution.
In this paper, we propose \textsf{MTCSC}, the constraint-based method for cleaning multivariate time series.
We formalize the repair problem, propose a linear-time method to employ online computing, and improve it by exploiting data trends.
We also support adaptive speed constraint capturing.
We analyze the properties of our proposals and compare them with SOTA methods in terms of effectiveness, efficiency versus error rates, data sizes, and applications such as classification.
Experiments on real datasets show that \textsf{MTCSC} can have higher repair accuracy with less time consumption.
Interestingly, it can be effective even when there are only weak or no correlations between the dimensions.
\end{abstract}
%
%

\settopmatter{printfolios=true}
\maketitle


\section{Introduction}
\label{sect:introduction}

Outliers are common in time series and can have two different meanings (errors or anomalies) \cite{DBLP:journals/csur/Blazquez-Garcia21}.
We focus on errors that result from measurement inaccuracies, data collection or system malfunctions \cite{aggarwal2016outlier}.
Even in areas such as stocks and flights, a surprisingly large amount of inconsistent data is observed \cite{DBLP:journals/pvldb/LiDLMS12}.
It is important to identify and correct these erroneous values to ensure the accuracy and reliability of time series analyzes \cite{hyndman2018forecasting}.

\subsection{Motivation}
Constraint-based methods for cleaning time series have been proposed in recent years.
Existing studies \cite{DBLP:conf/sigmod/SongZWY15, DBLP:journals/tods/SongGZWY21} consider the constraints of speeds and accelerations on value changes, namely speed/ acceleration constraints, respectively.
They recognize the violations and modify these dirty points according to the minimum change principle \cite{DBLP:conf/sigmod/BohannonFFR05} to the repaired results that correspond to the defined constraints.
However, such a principle leads to maximum/minimum compatible values being the final results (referred to as border repair), which can significantly change the data distribution \cite{DBLP:journals/pvldb/DasuL12}.
On the other hand, current constraint-based methods all focus on univariate time series, i.e. with only one dimension.
A multivariate data point can meet the given constraint in each of the single dimensions, but still violate the speed constraint in the multivariate case (see data point at $\mathit{t}_8$ in Example \ref{example:motivation}).
Furthermore, they cannot recognize the small errors that actually meet the speed constraints.
The small errors are very important in some applications such as automatic driving \cite{DBLP:conf/sigmod/ZhangSW16}.

Smoothing methods such as \cite{GARDNER2006637} use the weighted average of previous values as repaired results and suffer from the problem of over-repair, where almost all data points are changed even though most of them are originally correct.
The statistical-based method \cite{DBLP:conf/sigmod/ZhangSW16} builds a probability distribution model of speed changes between neighboring data points over the entire data to repair small errors.
\cite{DBLP:journals/vldb/WangZSW24} develops algorithms in streaming scenarios using a dynamic probability construction.
Although small errors are successfully repaired, the above methods are still proposed for univariate time series and do not capture the correlations between data dimensions.
The more recent learning-based methods \cite{DBLP:journals/pvldb/TuliCJ22, DBLP:journals/corr/abs-2107-12626} pay more attention to the anomalies (not errors), i.e. the observations that deviate significantly from the majority of the data.
These methods always require a large number of clean data samples and are quite sensitive to the hyper-parameters.



\begin{figure}[t]
  \centering
  \includegraphics[width=\figwidths]{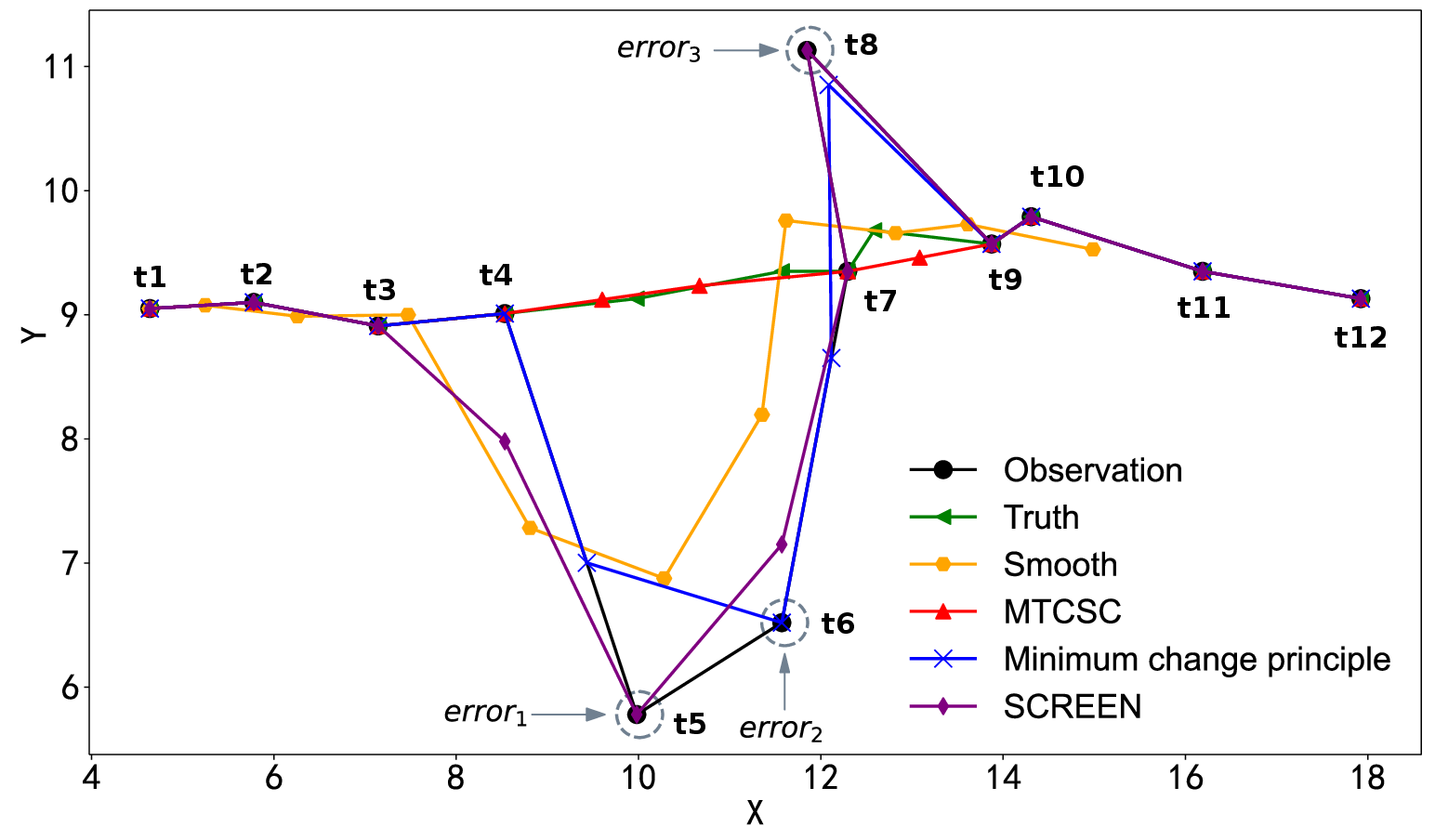}
  \caption{Example of observations and different repairs}
  \label{fig:example1}
\end{figure}

\begin{example}
\label{example:motivation}
Consider a continuous sequence of 12 GPS trajectories in Figure \ref{fig:example1}.
There are considerable deviations in the observations (in black) from $\mathit{t}_5, \mathit{t}_6, \mathit{t}_8$, which are influenced by the passage through the building and are regarded as errors.

The smoothing method (in orange) changes almost all points, even the originally correct ones.
As a univariate method, \textsf{SCREEN} considers each dimension separately.
The point at $\mathit{t}_8$ satisfies the speed constraint in every single dimension and remains unchanged, the point at $\mathit{t}_4$ is changed accordingly (in purple).

Then we consider using speed constraint in the multivariate case.
Following the minimum change principle, the border repair is applied at $\mathit{t}_5$ and $\mathit{t}_8$, and
$\mathit{t}_6$ is therefore considered correct.
If we take advantage of the data distribution (trend of succeeding points) and manage to change the least number of points, the repair result (our proposal \textsf{MTCSC}, in red) is closer to the ground truth (in green).
\end{example}

\subsection{Solution}
\label{sect:solution}
To overcome the above problems of the existing methods, we apply the speed constraint over all dimensions together instead of considering each dimension separately.
Moreover, we do not follow the \emph{minimum change principle} \cite{DBLP:conf/sigmod/BohannonFFR05} to minimize the total distance before and after the repair,
\begin{revised}
since a certain portion of originally clean data is still changed.
\end{revised}
Instead, we alternatively manage to fix the minimum number of data points and maintain the data distribution, a.k.a. the \emph{minimum fix principle} \cite{DBLP:conf/icdt/AfratiK09}.
For small errors that are difficult to detect by the speed constraint, we use the information about the data distribution in the sliding window and modify the points even if they satisfy the given constraints.
Based on the above techniques, we propose \textsf{MTCSC} to clean the multivariate time series under speed constraints.

However, there still exists some challenges.
(1) It may be difficult to explain why it makes sense to consider the speed constraint across all dimensions together.
For a GPS trajectory, it is obvious that the speed should be calculated across all two dimensions.
But it is not convincing enough to calculate the speed across the dimensions without any correlations.
Will this work in practice and why?
(2) When considering the multivariate case, the search for a suitable repair candidate is not trivial due to the \emph{curse of dimensionality} \cite{taylor2019dynamic}.
On the one hand, the higher the dimensions, the sparser the data points are, which makes it difficult to set a suitable speed constraint.
On the other hand, the candidate space becomes quite large.
In the univariate case, the repair candidate lies on a line, in the two-dimensional case on a circle, in the three-dimensional case on a cube, and so on.
(3) Time series are not stationary \cite{DBLP:journals/jips/DingZLWWL23} in some scenarios, but the success of our proposal depends on the correct specification of the speed constraint.
How to dynamically determine the speed constraint when points arrive becomes another challenge.

\subsection{Contribution}
The proposed \textsf{MTCSC} is a linear time, constant space cleaning method over multivariate time series.
Our major contributions in this paper are summarized as:
\begin{enumerate}[fullwidth]
\item
We formalize the repair problem over multivariate time series under speed constraints by considering the whole time series or data points in a window, respectively, in Section \ref{sect:problem}.
By transforming to the MIQP/MILP problem, existing solver can be employed.
We further reduce the time complexity to polynomial time by using the idea of dynamic programming (\textsf{MTCSC-G}).
\item 
We design an online linear time cleaning method (\textsf{MTCSC-L}) in Section \ref{sect:local-stream} to locally determine whether the first point (key point) in the current window to be repaired or not as data point arrives. 
It is notable that soundness w.r.t. speed constraint satisfaction is guaranteed in the devised algorithm.
\item
We enhance the online method by capturing the data distribution in the local window via clustering (\textsf{MTCSC-C}) in Section \ref{sect:cluster} to further promote the accuracy. 
Considering the trend of the following data points in the given window, small errors that satisfy the speed constraint can also be repaired.
\item
We present an adaptive method (\textsf{MTCSC-A}) that can capture the proper speed constraint in Section \ref{sect:adaptive}. 
Once the difference of the speed distribution in a certain period is larger than a per-defined threshold, which means the characteristic of the series has been changed, the speed constraint is then updated to capture the feature of the incoming data.
\item
We analyze the experimental results on real datasets and discuss the superiority and limitations of our proposed methods in Section \ref{sect:experiment}. 
\end{enumerate}

Table \ref{table:notations} lists the notations frequently used in this paper.

\begin{table}[ht]
 \caption{Notations}
 \label{table:notations}
 \centering
 \resizebox{0.9\expwidths}{!}{%
 \begin{tabular}{rp{2.5in}}
  \toprule
  Symbol & Description \\ 
  \midrule
  $\boldsymbol{\mathit{x}}$ & time series \\
  $\boldsymbol{\mathit{x}}_i$, $\boldsymbol{\mathit{x}}[i]$ & $i$-th data point in $\boldsymbol{\mathit{x}}$ \\
  $\mathit{t}_i$ & timestamp of $i$-th data point in $\boldsymbol{\mathit{x}}$\\
  $\mathit{s}$ & speed constraint \\
  $\mathit{w}$ & window size of speed constraint \\
  $n$ & length of time series $\boldsymbol{\mathit{x}}$ \\
  $D$ & dimension of time series $\boldsymbol{\mathit{x}}$ \\
  $\boldsymbol{\mathit{x}}'$ & repair of time series $\boldsymbol{\mathit{x}}$ \\
  $\mathit{z}_i$ & 0 if $\boldsymbol{\mathit{x}}_i$ is not modified, otherwise, 1 \\
  $M$ & a constant number which is sufficient large \\
  \textsf{satisfy}$(\boldsymbol{\mathit{x}}_i, \boldsymbol{\mathit{x}}_j)$ & $\boldsymbol{\mathit{x}}_i$ is compatible with $\boldsymbol{\mathit{x}}_j$ w.r.t. $\mathit{s}$ \\
 \bottomrule
 \end{tabular}
}
\end{table}


\section{Problem Statement}
\label{sect:problem}

In this section, we first introduce the time series and the speed constraint that are considered in our proposal.
Then, we define the repair problem using the \emph{minimum fix principle} with all data points as a whole.
Finally, we formalize the problem when we consider the constraints only locally in a given window.

\subsection{Speed Constraint}
\label{sect:constraint}

\begin{definition}[Time Series]
\label{definition:timeseries}
A time series is a sequence of data points indexed in time order. 
Consider a time series of $n$ observations, $\boldsymbol{\mathit{x}} = \{\boldsymbol{\mathit{x}}_1, \ldots, \boldsymbol{\mathit{x}}_n\}$. 
The $i$-th observation (data point) $\boldsymbol{\mathit{x}}_i \in \mathbb{R}^{D}$ consists of $D$ dimensions $\{\mathit{x}^1_i, \ldots, \mathit{x}^D_i\}$ and is observed at time $\mathit{t}_i$.
\end{definition}

\begin{definition}[Distance]
The distance between two data points is the Euclidean distance between them, denoted by 
$d(\boldsymbol{\mathit{x}}_i, \boldsymbol{\mathit{x}}_j) = \sqrt{\sum_{\ell=1}^{D}(\mathit{x}^{\ell}_i-\mathit{x}^{\ell}_j)^2}$.
\end{definition}

\begin{definition}[Speed Constraint]
\label{definition:speed-constraint}
A \emph{speed constraint} $\mathit{s}=(\mathit{s}_{\min},\allowbreak \mathit{s}_{\max})$ with window size $\mathit{w}$ is a pair of minimum speed $\mathit{s}_{\min}$ and maximum speed $\mathit{s}_{\max}$ over the time series $\boldsymbol{\mathit{x}}$.
We say that a time series $\boldsymbol{\mathit{x}}$ \emph{satisfies} the speed constraint $\mathit{s}$, denoted by $\boldsymbol{\mathit{x}}\vDash\mathit{s}$, if for any $\boldsymbol{\mathit{x}}_i, \boldsymbol{\mathit{x}}_j$ in a window, i.e.,  
$0 < \mathit{t}_j - \mathit{t}_i \leq \mathit{w}$, it has 
$\mathit{s}_{\min}\leq\frac{d(\boldsymbol{\mathit{x}}_i, \boldsymbol{\mathit{x}}_j)}{\mathit{t}_j-\mathit{t}_i}\leq\mathit{s}_{\max}$, 
also denoted by \textsf{satisfy}($\boldsymbol{\mathit{x}}_i, \boldsymbol{\mathit{x}}_j$),
where window $\mathit{w}$ denotes a period of time.
\end{definition}

As mentioned in \cite{DBLP:conf/sigmod/SongZWY15}, speed constraint is often useful in real-world scenarios within a certain time period.
For example, it makes sense to consider the maximum walking/running/driving speed in hours, as a person cannot usually perform these activities for days or longer without interruption.
Stock prices are limited to a certain range based on the last trading day's prices.
\begin{revised}
Speed constraint is applicable in practice.
It supports online cleaning and takes less time while maintaining relatively good repair accuracy.
Time series is naturally infinite.
An online method can better meet these requirements.
Meanwhile, speed constraint can be efficiently captured dynamically, which accommodates the non-stationary property of time series.
\end{revised}
Moreover, it is easy to know that $\mathit{s}_{\min} = 0$ in real cases because the minimum Euclidean distance will always be $0$.
Therefore, we assume that $\mathit{s}_{\min} = 0$ and refer to $\mathit{s}_{\max}$ as $\mathit{s}$ in the rest of this article.

Another observation is that the ``jump'' of values in time series data within a short time period is usually limited and has been verified for univariate data (i.e. only one dimension) \cite{DBLP:conf/sigmod/SongZWY15, DBLP:journals/tods/SongGZWY21}.
As for the multivariate case, one can extend this observation to data that exhibit high correlations over the distance between their dimensions.
For example, the GPS trajectory data.
However, it is surprising that our proposed speed constraint over multiple dimensions can also be effective for data with weak or no correlations.
See Section \ref{sect:exp-correlation} for an experimental analysis.

A repair $\boldsymbol{\mathit{x}}'$ of $\boldsymbol{\mathit{x}}$ is a modification of the values in $\boldsymbol{\mathit{x}}_i$ to $\boldsymbol{\mathit{x}}_i'$, 
i.e., $\exists\ell\in(1,D), \mathit{x}^{\ell}_i \neq \mathit{x}^{\ell'}_i$, where $\mathit{t}_i' = \mathit{t}_i$.
Referring to the minimum fix principle, the repair number is counted as 

\begin{equation}
\label{equation:repair-cost}
\Delta(\boldsymbol{\mathit{x}}', \boldsymbol{\mathit{x}}) = \sum_{i=1}^{n}\mathbb{I}(\boldsymbol{\mathit{x}}_i' \neq \boldsymbol{\mathit{x}}_i) = \sum_{i=1}^{n}\mathit{z}_i
\end{equation}

\begin{figure}[t]
  \centering
  \includegraphics[width=\figwidths]{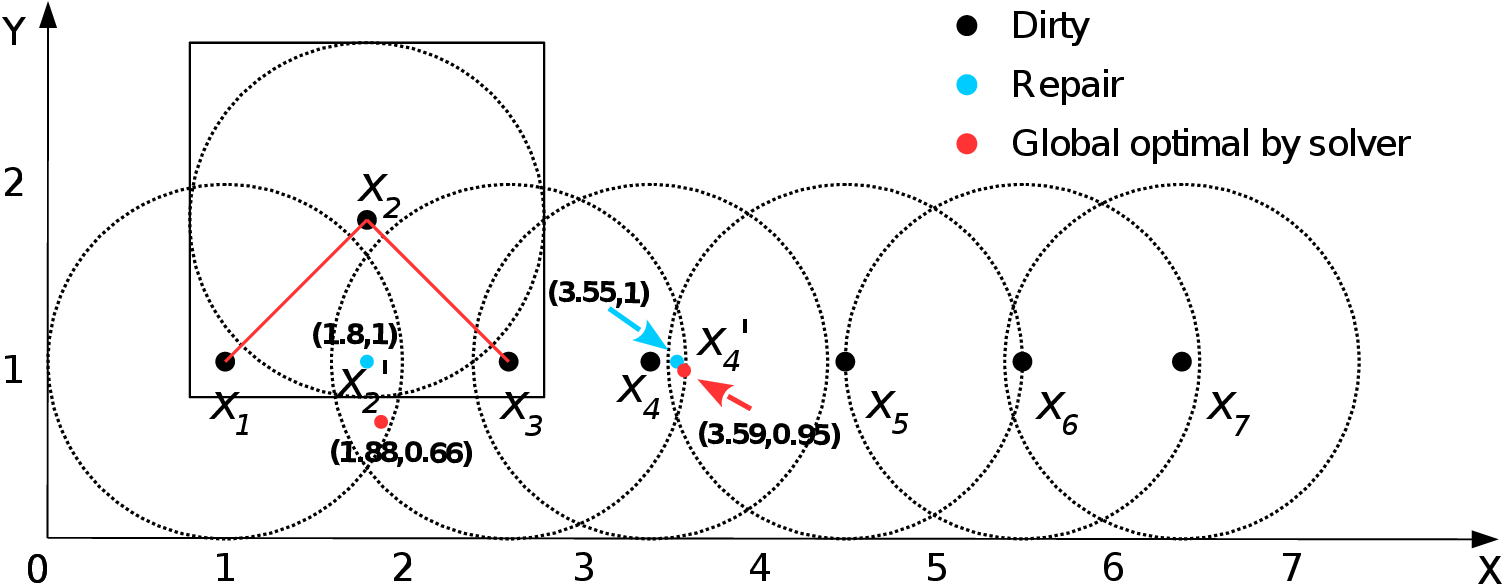}
  \caption{Possible repairs under speed constraints}%
  \label{fig:example2}
\end{figure}

\begin{example}[Speed constraints and repairs]
\label{example:dirty}
Consider a time series ($D = 2$) $\boldsymbol{\mathit{x}} = \{(1,1), (1.8,1.8), (2.6,1), (3.4,1), (4.5,1), (5.5,1), \allowbreak (6.4,1)\}$ of 7 points, with timestamps $\mathit{t} = \{1, 2, \allowbreak 3, 4, 5, 6, 7\}$.
Let the speed constraint $\mathit{s} = 1$.
Figure \ref{fig:example2} illustrates the data points (in black).
The dashed circles in the figure all have a radius of $1 (= \mathit{s} \cdot 1)$.
If data points from adjacent timestamps fall into each other's dashed circles, this means that they mutually satisfy the speed constraints.

Let windows size $\mathit{w} = 1$.
Points $\boldsymbol{\mathit{x}}_1$ and $\boldsymbol{\mathit{x}}_2$, with time interval  $2-1 \leq 1$ in a window, are identified as violations to $\mathit{s}$, since the speed is $\sqrt{(1.8-1)^2+(1.8-1)^2}/(2-1) \approx 1.13>1$.
Similarly, $\boldsymbol{\mathit{x}}_2$ and $\boldsymbol{\mathit{x}}_3$ with speed $\sqrt{(2.6-1.8)^2+(1-1.8)^2}/(2-1) \approx 1.13>1$ are violations.
However, if the two dimensions are considered separately, $\boldsymbol{\mathit{x}}_2$ will be compatible with $\boldsymbol{\mathit{x}}_1$ and $\boldsymbol{\mathit{x}}_3$ (denoted by black square, the speed are both 0.8 in x and y direction).

To fix the violations (indicated by red lines), a repair can be performed on $\boldsymbol{\mathit{x}}_2$, i.e.,$\boldsymbol{\mathit{x}}'_2$ = (1.8,1) (the blue point).
As shown in Figure \ref{fig:example1}, $\boldsymbol{\mathit{x}}'_2$ is compatible with $\boldsymbol{\mathit{x}}_1$ and $\boldsymbol{\mathit{x}}_3$.
Similarly, a repair can be performed for $\boldsymbol{\mathit{x}}_4$, i.e. $\boldsymbol{\mathit{x}}'_4$ = (3.55,1) (the blue point).
The repaired sequence therefore satisfies the speed constraints.
The repaired points are $\boldsymbol{\mathit{x}}_2$ and $\boldsymbol{\mathit{x}}_4$, and the repair number is 2.
\end{example}

\subsection{Global Optimal}
\label{sect:global}

The cleaning problem is to find a repair time series, termed as \emph{global optimal}, that satisfies the given speed constraints and minimizes the number of points changed from the original time series.

\begin{problem}
Given a finite time series $\boldsymbol{\mathit{x}}$ of $n$ data points and a speed constraint $\mathit{s}$, the global optimal repair problem is to find a repair $\boldsymbol{\mathit{x}}'$ such that $\boldsymbol{\mathit{x}}'\vDash\mathit{s}$ and $\Delta(\boldsymbol{\mathit{x}}', \boldsymbol{\mathit{x}})$ is minimized.
\end{problem}

The global optimal repair problem can be formalized as follows, 
\begin{align}
\label{equation:global}
\min\quad & \sum_{i=1}^{n}\mathit{z}_i, & & \mathit{z}_i \in \{0, 1\} \\ 
\label{equation:global-max}
\mathrm{s.t.}\quad & 
\frac{d(\boldsymbol{\mathit{x}}'_i, \boldsymbol{\mathit{x}}'_j)}{\mathit{t}_j-\mathit{t}_i} \leq\mathit{s}, &  & 
\mathit{t}_i<\mathit{t}_j\leq\mathit{t}_i+\mathit{w}, 1\leq i,j\leq n \\ 
\label{equation:global-repair}
 &
\mathit{x}^{\ell'}_i \in [\mathit{x}^{\ell}_i \pm M\cdot\mathit{z}_i], &  &
M>c, 1\leq i\leq n, 1\leq\ell\leq D
\end{align}
where $c$ is constant integer that large enough to make $\mathit{x}^{\ell}_i$ feasible.

The correctness of result $\mathit{x}'$ is easy to prove. 
Formula (\ref{equation:global}) is the repair number in formula (\ref{equation:repair-cost}) to minimize. 
The speed constraints are ensured in formula (\ref{equation:global-max}), 
by considering all the $\mathit{t}_j$ in the window starting from $\mathit{t}_i$.
Finally, formula (\ref{equation:global-repair}) identifies the repair value for each data point $i$ in the time series.

\subsubsection{Optimization Solver}
\label{sect:global-solver}
After formalization, the problem is a standard mixed-integer quadratic program (MIQP) problem \cite{burer2011milp},
and $\mathit{z}_i, \mathit{x}^{\ell'}_i, 1\leq i\leq n, 1\leq\ell\leq D$ are variables in problem solving. 
If the input is univariate, the problem is reduced to a mixed-integer linear program (MILP) problem \cite{DBLP:books/fm/GareyJ79}, since the distance is linear.
Existing solvers like Gurobi \cite{Gurobi2024} can be used directly.
\begin{revised}
Once the input time series has been converted to the constraint type, according to the formulas (\ref{equation:global-max}) and (\ref{equation:global-repair}), where M is set to 1000, it can be fed directly into an existing optimization solver.
The output of the solver is the repair values of each data points and can be the ``optimal'' solution under our problem definition.
However, the choice of vanilla output as repaired results can still be improved, as we only minimize the number of repaired points but do not consider the magnitude of the changes, which can lead to a repair that is too far from the original point and less accurate.
\end{revised}

\begin{example}[Global optimal by solver, Example \ref{example:dirty} continued]
\label{example:optimization-solver}
Consider again the time series $\boldsymbol{\mathit{x}}$ and the speed constraint $\mathit{s} = 1$ in Example \ref{example:dirty} with window size $\mathit{w} = 7$, as illustrated in Figure \ref{fig:example1}.
\\
According to formula (\ref{equation:global-max}), the constraint declared w.r.t. $\mathit{s}$ are:
\begin{align*} &
\frac{d(\boldsymbol{\mathit{x}}'_1, \boldsymbol{\mathit{x}}'_2)}{2-1} \leq 1, & &
\frac{d(\boldsymbol{\mathit{x}}'_1, \boldsymbol{\mathit{x}}'_3)}{3-1} \leq 1, & &
\frac{d(\boldsymbol{\mathit{x}}'_1, \boldsymbol{\mathit{x}}'_4)}{4-1} \leq 1, & &
\frac{d(\boldsymbol{\mathit{x}}'_1, \boldsymbol{\mathit{x}}'_5)}{5-1} \leq 1, \\ &
\ldots \\ &
\frac{d(\boldsymbol{\mathit{x}}'_4, \boldsymbol{\mathit{x}}'_7)}{7-4} \leq 1, & &
\frac{d(\boldsymbol{\mathit{x}}'_5, \boldsymbol{\mathit{x}}'_6)}{6-5} \leq 1, & &
\frac{d(\boldsymbol{\mathit{x}}'_5, \boldsymbol{\mathit{x}}'_7)}{7-5} \leq 1, & &
\frac{d(\boldsymbol{\mathit{x}}'_6, \boldsymbol{\mathit{x}}'_7)}{7-6} \leq 1. 
\end{align*}
We use Gurobi Optimizer to get the global optimal solution with the minimum fix is $\boldsymbol{\mathit{x}}'_2=(1.88,0.66), \boldsymbol{\mathit{x}}'_4=(3.59,0.95)$, as shown by the red points in Figure \ref{fig:example2}.
The repair number is 2.
\end{example}

\subsubsection{Dynamic Programming}
\label{sect:global-dp}

\begin{algorithm}[!ht]
\caption{\textsf{MTCSC-G}($\boldsymbol{\mathit{x}}, \mathit{s}$)}
\label{alg:global}
  \KwIn{time series $\boldsymbol{\mathit{x}}$, speed constraint $\mathit{s}$}
  \KwOut{index of points that need to be fixed $FixList$}
    $dp[1..n] \leftarrow \{1\}$, $preIndex[1..n] \leftarrow \{0\}$\;
    $maxLength \leftarrow 0$, $endIndex \leftarrow 0$\;
    
    \For{$i \leftarrow 0$ to $n-1 \do$}{
      \For{$j \leftarrow 0$ to $i-1 \do$}{
        \If{\textsf{satisfy}$(\boldsymbol{\mathit{x}}_i, \boldsymbol{\mathit{x}}_j)$ \textbf{and} $dp[i] < dp[j] + 1$}{
          $dp[i] \leftarrow dp[j] \ + \ 1$\;
          $preIndex[i] \leftarrow j$\;
        }
      }
      \If{$dp[i] > maxLength$}{
        $maxLength \leftarrow dp[i]$\;
        $endIndex \leftarrow i$\;
      }
    }
    \For{$k \leftarrow 1$ to $maxLength$ \do}{
      $CleanList.add(endIndex)$\;
      $endIndex \leftarrow preIndex[endIndex]$\;
    }
    $FixList \leftarrow \{1,\ldots n\} - CleanList$\;
    \Return $FixList$\;
\end{algorithm}

\begin{revised}
The optimization solver can lead to a global optimal solution, but it costs too much time.
Therefore, we propose Algorithm \ref{alg:global} to determine which data points need to be fixed to improve efficiency.
It is an extension of the longest increasing subsequence problem (LIS) \cite{schensted1961longest}.
Our goal is to find the longest subsequence in the entire time series that satisfies the given speed constraint, which is equivalent to the LIS problem of finding the longest increasing subsequence.
Once it is found (i.e. \emph{CleanList}), the data points that are not included in this list must be repaired and are guaranteed to be the least modified (i.e. minimum fix).
Finally, for each point in \emph{FixList}, we can easily find the nearest preceding point $\boldsymbol{\mathit{x}}_p$ and the nearest succeeding point $\boldsymbol{\mathit{x}}_m$ in \emph{CleanList}.
The interpolation of these two points according to the formula (\ref{equation:local-stream}) is used as the final repair result.
The correctness can be verified in Proposition \ref{the:pre-range} and Proposition \ref{the:post-range} in Section \ref{sect:stream}. 
As for the time analysis, we know that the judgment procedure \textsf{satisfy}$(\boldsymbol{\mathit{x}}_i, \boldsymbol{\mathit{x}}_j)$ will cost $O(D)$ time where D is the dimension of time series.
Algorithm \ref{alg:global} outputs the \emph{FixList}, which is in fact the detection of violations and the detection time is $O(Dn^2)$.
The overall time complexity is $O(Dn^2) + O(fD) = O(Dn^2)$ where $f$ is the number of all the points need to be fixed and the worst-case scenario is $f = n/2$.
\end{revised}

\begin{example}[Global optimal by DP, Example \ref{example:dirty} continued]
\label{example:global-dp}
Consider again the time series $\boldsymbol{\mathit{x}}$ and the speed constraint $\mathit{s} = 1$ in Example \ref{example:dirty} with window size $\mathit{w} = 7$.
We use Algorithm \ref{alg:global} and get the FixList is \{2,4\}.
This means that the minimum fix repair requires the repair of $\boldsymbol{\mathit{x}}_2$ and $\boldsymbol{\mathit{x}}_4$.
The clean list is therefore \{1,3,5,6,7\}.
We use the ``clean'' points (considered clean by Algorithm \ref{alg:global}) closest to the detected points on both sides for repair.
So we use $\boldsymbol{\mathit{x}}_1$ and $\boldsymbol{\mathit{x}}_3$ to repair $\boldsymbol{\mathit{x}}_2$, similarly, we use $\boldsymbol{\mathit{x}}_3$ and $\boldsymbol{\mathit{x}}_5$ to repair $\boldsymbol{\mathit{x}}_4$.
The final repaired points are:
\begin{align*} &
\mathit{x}^{1'}_2 = \frac{\mathit{t}_2-\mathit{t}_1}{\mathit{t}_3-\mathit{t}_1}  (2.6-1) + 1 = 1.8, & &
\mathit{x}^{2'}_2 = \frac{\mathit{t}_2-\mathit{t}_1}{\mathit{t}_3-\mathit{t}_1}  (1-1) + 1 = 1, \\ &
\mathit{x}^{1'}_4 = \frac{\mathit{t}_4-\mathit{t}_3}{\mathit{t}_5-\mathit{t}_3}  (4.5-2.6) + 2.6 = 3.55, & &
\mathit{x}^{2'}_4 = \frac{\mathit{t}_4-\mathit{t}_3}{\mathit{t}_5-\mathit{t}_3}  (1-1) + 1 = 1.
\end{align*}
The global optimal solution with Algorithm \ref{alg:global} with the minimum fix is $\boldsymbol{\mathit{x}}'_2=(1.8,1), \boldsymbol{\mathit{x}}'_4=(3.55,1)$, shown as blue points in Figure \ref{fig:example2}.
The repair number is 2, which matches the results of the solver.
\end{example}

\subsection{Local Optimal}
\label{sect:local}

The global optimal requires the entire time series as input and cannot support online cleaning.
Therefore, we first study the local optimal, which concerns only the constraints locally in a window.
We say a data point $\boldsymbol{\mathit{x}}_k$ \emph{locally satisfies} the speed constraint $\mathit{s}$, denoted by $\boldsymbol{\mathit{x}}_k\vDash\mathit{s}$, 
if for any $\boldsymbol{\mathit{x}}_i$ in the window starting from $\boldsymbol{\mathit{x}}_k$, 
i.e., $\mathit{t}_k<\mathit{t}_i\leq\mathit{t}_k+\mathit{w}$, it has 
$\textsf{satisfy}(\boldsymbol{\mathit{x}}'_i, \boldsymbol{\mathit{x}}'_k)$.

\begin{problem}
Given a data point $\boldsymbol{\mathit{x}}_k$ in a time series $\boldsymbol{\mathit{x}}$ and a speed constraint $\mathit{s}$,
the local optimal repair problem is to find a repair $\boldsymbol{\mathit{x}}'$ such that $\boldsymbol{\mathit{x}}'_k\vDash\mathit{s}$ and $\Delta(\boldsymbol{\mathit{x}}', \boldsymbol{\mathit{x}})$ is minimized.
\end{problem}

We formalize the local optimal repair problem as
\begin{align}
\label{equation:local}
\min\quad & \sum_{i=1}^{n}\mathit{z}_i, & & \mathit{z}_i \in \{0, 1\} 
\\ \nonumber
\mathrm{s.t.}\quad & 
\frac{d(\boldsymbol{\mathit{x}}'_k,\boldsymbol{\mathit{x}}'_i)}{\mathit{t}_k-\mathit{t}_i} \leq\mathit{s}, &  & 
\mathit{t}_k<\mathit{t}_i\leq\mathit{t}_k+\mathit{w},  
 1\leq i\leq n 
\\ \nonumber
&
\mathit{x}^{\ell'}_i\in[\mathit{x}^{\ell}_i\pm M\cdot\mathit{z}_i], &  &
M>c, \mathit{t}_k<\mathit{t}_i\leq\mathit{t}_k+\mathit{w}, 1\leq\ell\leq D
\end{align}
Similarly, $\mathit{z}_i, \mathit{x}^{\ell'}_i, 1\leq i\leq n, 1\leq\ell\leq D$ are variables in problem solving in the existing solver.
We only change data points $i$ with $\mathit{t}_k\leq\mathit{t}_i\leq\mathit{t}_k+\mathit{w}$ in the window starting from the current $\boldsymbol{\mathit{x}}_k$ under the local optimal setting and therefore only have much fewer variables.
Nevertheless, the speed constraints are still satisfied in the given window.
Moreover, we can also develop a greedy algorithm whose time complexity is $O(n)$.
(See Section \ref{sect:local-stream} for more details.)

\begin{example}[Local optimal, Example \ref{example:dirty} continued]
Consider again the time series $\boldsymbol{\mathit{x}}$ and the speed constraint $\mathit{s} = 1$ in Example \ref{example:dirty} with window size $\mathit{w} = 7$, as illustrated in Figure \ref{fig:example1}.
\\
Let $\mathit{t}_k=1$, i.e. $\mathit{x}_1=(1,1)$ be the currently considered data point.
According to formula (\ref{equation:local}), the constraint for the local optimal on $\boldsymbol{\mathit{x}}_1$ are:
\begin{align*} &
\frac{d(\boldsymbol{\mathit{x}}'_1, \boldsymbol{\mathit{x}}'_2)}{2-1} \leq 1, & &
\frac{d(\boldsymbol{\mathit{x}}'_1, \boldsymbol{\mathit{x}}'_3)}{3-1} \leq 1, & &
\frac{d(\boldsymbol{\mathit{x}}'_1, \boldsymbol{\mathit{x}}'_4)}{4-1} \leq 1, \\ &
\frac{d(\boldsymbol{\mathit{x}}'_1, \boldsymbol{\mathit{x}}'_5)}{5-1} \leq 1, & &
\frac{d(\boldsymbol{\mathit{x}}'_1, \boldsymbol{\mathit{x}}'_6)}{6-1} \leq 1, & &
\frac{d(\boldsymbol{\mathit{x}}'_1, \boldsymbol{\mathit{x}}'_7)}{7-1} \leq 1.
\end{align*}
In the local optimal, only $\frac{d(\boldsymbol{\mathit{x}}_1, \boldsymbol{\mathit{x}}_2)}{2-1} \approx 1.33 > 1$ violates the speed constraint, 
while $\frac{d(\boldsymbol{\mathit{x}}_1, \boldsymbol{\mathit{x}}_4)}{4-1} = 2.4 < 3$ satisfies the speed constraint.
So we just need to repair $\boldsymbol{\mathit{x}}_2$.
We use $\boldsymbol{\mathit{x}}_k$, i.e. $\boldsymbol{\mathit{x}}_1$ and the first data point (i.e. $\boldsymbol{\mathit{x}}_3$) after $\boldsymbol{\mathit{x}}_2$ that satisfies the speed constraint with  $\boldsymbol{\mathit{x}}_1$ to repair $\boldsymbol{\mathit{x}}_2$.
\begin{align*} &
\mathit{x}^{1'}_2 = \frac{\mathit{t}_2-\mathit{t}_1}{\mathit{t}_3-\mathit{t}_1} \cdot (2.6-1) + 1 = 1.8, & &
\mathit{x}^{2'}_2 = \frac{\mathit{t}_2-\mathit{t}_1}{\mathit{t}_3-\mathit{t}_1} \cdot (1-1) + 1 = 1. & &
\end{align*}
The local optimal solution with the minimum fix is $\mathit{x}'_2=(1.8,1)$.
The repair number is 1, which is less than the global optimal because the constraints of the local optimal are fewer than those of the global optimal.
\end{example}


\section{Streaming Computation}
\label{sect:stream}

A time series can also have an infinite number of data points, i.e. it arrives continuously.
To support online cleaning, we first propose \textsf{MTCSC-L} w.r.t. local optimal by deciding whether the current data point should be fixed or not in Section \ref{sect:local-stream}.
To support cleaning errors that satisfy the speed constraint, we further devise \textsf{MTCSC-C}, which takes into account the trend of data points in the current window in Section \ref{sect:cluster}.
By sliding windows in the time series, both the result of \textsf{MTCSC-L} $\mathit{x}_\text{local}$ and \textsf{MTCSC-C} $\mathit{x}_\text{clu}$ guarantee to satisfy the given speed constraints.
Since the global optimal provides a minimum fixed repair $\mathit{x}_\text{global}$ over the entire data, we always have
$\Delta(\boldsymbol{\mathit{x}}, \mathit{x}_\text{global}) \leq \Delta(\boldsymbol{\mathit{x}}, \mathit{x}_\text{local})$ and
$\Delta(\boldsymbol{\mathit{x}}, \mathit{x}_\text{global}) \leq \Delta(\boldsymbol{\mathit{x}}, \mathit{x}_\text{clu})$.
Unfortunately, we do not have a theoretical upper bound of local and cluster fix number compared to the global ones. In practice, the fix numbers are very close to each other (as shown in Figure \ref{exp:stock}(d)).

\subsection{Local Streaming}
\label{sect:local-stream}

The proposed \textsf{MTCSC-L} is to iteratively determine the local optimal $\boldsymbol{\mathit{x}}'_k$, for $k=1,2,\dots$.
Consider $\boldsymbol{\mathit{x}}_k$, where $\boldsymbol{\mathit{x}}'_1,\dots,\boldsymbol{\mathit{x}}'_{k-1}$ have been determined before.
Referring to the speed constraints, each fixed $\boldsymbol{\mathit{x}}'_{j}, 
\mathit{t}_k-\mathit{w}\leq\mathit{t}_j<\mathit{t}_k,
1\leq j<k,$ provides a range of candidates for $\boldsymbol{\mathit{x}}'_k$, i.e., 
$\{(\mathit{x}^{1'}_k, \ldots, \mathit{x}^{D'}_k) \mid d(\boldsymbol{\mathit{x}}'_k, \boldsymbol{\mathit{x}}'_j)\leq\mathit{s}(\mathit{t}_k-\mathit{t}_j)\}$.
%
The following proposition states that considering the last $\boldsymbol{\mathit{x}}'_{k-1}$ is sufficient to determine the range of possible repairs for $\boldsymbol{\mathit{x}}'_{k}$.
The rationale is that for any $1\leq j<i<k$, $\boldsymbol{\mathit{x}}'_i$ should be in the range specified by $\boldsymbol{\mathit{x}}'_j$ as well. 
In other words, the candidate range of $\mathit{x}'_{k}$ specified by $\mathit{x}'_i$ is subsumed in the range by $\mathit{x}'_j$. 

\begin{proposition}
\label{the:pre-range}
For any $1\leq j<i<k, \mathit{t}_k-\mathit{w}\leq\mathit{t}_j<\mathit{t}_i<\mathit{t}_k$, we have 
$
\{\boldsymbol{\mathit{x}}'_k \mid d(\boldsymbol{\mathit{x}}'_k, \boldsymbol{\mathit{x}}'_i)\leq\mathit{s}(\mathit{t}_k-\mathit{t}_i)\}
\subset
\{\boldsymbol{\mathit{x}}'_k \mid d(\boldsymbol{\mathit{x}}'_k, \boldsymbol{\mathit{x}}'_j)\leq\mathit{s}(\mathit{t}_k-\mathit{t}_j)\}.
$
\end{proposition}

\begin{proof}
W.l.o.g. we take the two dimensional data $\boldsymbol{\mathit{x}}_k = \{(\mathit{x}_k, \mathit{y}_k)\}$ for all the proofs.
Since $\boldsymbol{\mathit{x}}'_i$ and $\boldsymbol{\mathit{x}}'_j$ are all fixed and in the same window, we have 
$d(\boldsymbol{\mathit{x}}'_i, \boldsymbol{\mathit{x}}'_j)\leq\mathit{s}(\mathit{t}_i-\mathit{t}_j)$.
The range specified by
$\boldsymbol{\mathit{x}}_k' = \{(\mathit{x}, \mathit{y}) \mid \sqrt{(\mathit{x}-\mathit{x}'_i)^2+(\mathit{y}-\mathit{y}'_i)^2}\leq\mathit{s}(\mathit{t}_k-\mathit{t}_i)\}$, i.e., 
$d(\boldsymbol{\mathit{x}}'_k, \boldsymbol{\mathit{x}}'_i) \leq \mathit{s}(\mathit{t}_k-\mathit{t}_i)$.
Points $\boldsymbol{\mathit{x}}'_k,\boldsymbol{\mathit{x}}'_i,\boldsymbol{\mathit{x}}'_j$ can form a triangle or a line.
Following the triangle inequality \cite{khamsi2011introduction}, we have
$d(\boldsymbol{\mathit{x}}'_k, \boldsymbol{\mathit{x}}'_j) 
\leq d(\boldsymbol{\mathit{x}}'_k, \boldsymbol{\mathit{x}}'_i) + d(\boldsymbol{\mathit{x}}'_i, \boldsymbol{\mathit{x}}'_j) 
\leq \mathit{s}(\mathit{t}_k-\mathit{t}_i) + \mathit{s}(\mathit{t}_i-\mathit{t}_j) 
= \mathit{s}(\mathit{t}_k-\mathit{t}_j)$.
Hence, $\boldsymbol{\mathit{x}}'_k$ is in the range specified by $\boldsymbol{\mathit{x}}'_j$.
\end{proof}
We can therefore identify the repair $\boldsymbol{\mathit{x}}'_k$ in a tight range of candidates specified by $\boldsymbol{\mathit{x}}'_{k-1}$.
However, the above ranges are still too large even if the dimension of the data is two (a circle of radius $s(\mathit{t}_k-\mathit{t}_{k-1})$).
Fortunately, we find that the sequence of speed changes is stationary for a short period of time \cite{DBLP:journals/vldb/WangZSW24}.
Then we use the idea of interpolation to determine the local solution of $\boldsymbol{\mathit{x}}'_k$ as
\begin{equation}
\label{equation:local-stream}
\mathit{x}_k^{\ell'}=\alpha\cdot(\mathit{x}_{m}^{\ell}-\mathit{x}_{p}^{\ell'})+\mathit{x}_{p}^{\ell'}, \quad 1\leq\ell\leq D
\end{equation}
where $\alpha = (\mathit{t}_{k}-\mathit{t}_{p})/(\mathit{t}_{m}-\mathit{t}_{p})$, 
$p$ is the last point before the window, 
$k$ is the first point in the window (key point) and 
$m$ is a point after $k$ with \textsf{satisfy}$(\boldsymbol{\mathit{x}}'_p, \boldsymbol{\mathit{x}}_m)$. 
%
The following proposition states that such a solution satisfies the given speed constraint.
\begin{proposition}
\label{the:correctness}
For any $\ell=1,\ldots,D$,
$\mathit{x}_k^{\ell'}=\alpha\cdot(\mathit{x}_{m}^{\ell}-\mathit{x}_{p}^{\ell'})+\mathit{x}_{p}^{\ell'}$, 
we have 
\textsf{satisfy}$(\boldsymbol{\mathit{x}}'_p, \boldsymbol{\mathit{x}}'_k)$.
\end{proposition}

\begin{proof}
Let $\mathit{r}_1=\mathit{s}(\mathit{t}_k-\mathit{t}_p),\mathit{r}_3=\mathit{s}(\mathit{t}_m-\mathit{t}_p)$ and $\ell\in[1, D]$.
Since \textsf{satisfy}$(\boldsymbol{\mathit{x}}'_p, \boldsymbol{\mathit{x}}_m)$, 
we have $\mathit{x}^{\ell}_m\in[\mathit{x}^{\ell'}_p-\mathit{r}_3, \mathit{x}^{\ell'}_p+\mathit{r}_3]$.
\begin{align*}
\mathit{x}_k^{\ell'}
 &= \alpha\cdot(\mathit{x}_{m}^{\ell} - \mathit{x}_{p}^{\ell'})+\mathit{x}_{p}^{\ell'} 
 = \alpha\cdot\mathit{x}_{m}^{\ell} + (1-\alpha)\mathit{x}_{p}^{\ell'} \\
 &\geq \alpha\cdot(\mathit{x}^{\ell'}_p-\mathit{r}_3) + (1-\alpha)\mathit{x}_{p}^{\ell'} 
 = \mathit{x}_{p}^{\ell'} - \alpha\cdot\mathit{r}_3 \\
 &= \mathit{x}_{p}^{\ell'} - \mathit{s}(\mathit{t}_m-\mathit{t}_p)\cdot(\mathit{t}_{k}-\mathit{t}_{p})/(\mathit{t}_{m}-\mathit{t}_{p}) \\
 &= \mathit{x}_{p}^{\ell'} - \mathit{s}(\mathit{t}_{k}-\mathit{t}_{p}) 
 = \mathit{x}_{p}^{\ell'} - \mathit{r}_1
\end{align*}
Similarly, we have $\mathit{x}_k^{\ell'}\leq\mathit{x}^{\ell'}_p+\mathit{r}_1$.
Since point $\boldsymbol{\mathit{x}}'_k$ is in the same line with points $\boldsymbol{\mathit{x}}'_p$ and $\boldsymbol{\mathit{x}}_m$, hence, we have \textsf{satisfy}$(\boldsymbol{\mathit{x}}'_p, \boldsymbol{\mathit{x}}'_k)$.
\end{proof}

Algorithm \ref{alg:local} presents the online local repair of a time series $\boldsymbol{\mathit{x}}$ w.r.t.\ local optimal under the speed constraint $\mathit{s}$. 
\begin{revised}
It detects the violation for each data point in Line 2 and the time cost is $O(D)$.
\end{revised}
For each point $k$, $k=1,2,\dots,n$ that violates $\mathit{s}$ w.r.t.\ point $k-1$, 
Line 4 will find the first succeeding point $i$ that \textsf{satisfy}$(\boldsymbol{\mathit{x}}_i, \boldsymbol{\mathit{x}}'_{k-1})$.
Lines 5 and 6 indicate that if there is no point satisfying $\mathit{s}$ w.r.t.\ point $k-1$, 
$\boldsymbol{\mathit{x}}'_{k-1}$ will be determined as the repair of point $k$.
Otherwise, Lines 8 to 12 will set $\boldsymbol{\mathit{x}}'_{k}$ following formula (\ref{equation:local-stream}).
It is easy to know that the time complexity of Algorithm \ref{alg:local} is 
\begin{revised}
$O(wDn)$.
\end{revised}

\begin{algorithm}[!ht]
\caption{\textsf{MTCSC-L}($\boldsymbol{\mathit{x}}, \mathit{s}$)}
\label{alg:local}
  \KwIn{time series $\boldsymbol{\mathit{x}}$, speed constraint $\mathit{s}$}
  \KwOut{a repair $\boldsymbol{\mathit{x'}}$ of $\boldsymbol{\mathit{x}}$ w.r.t local optimal}
    \For{$k \leftarrow 1$ to $n$ \do}{
      \If{\textsf{satisfy}$(\boldsymbol{\mathit{x}}_k, \boldsymbol{\mathit{x}}'_{k-1})$ \textbf{or} $k=1$} {
        \Continue\;
	  }
	  \For{$i \leftarrow k+1$ to $n$ \do}{
	    \If{$t_{i} > t_{k}+\mathit{w}$}{
	      $\boldsymbol{\mathit{x}}_k' \leftarrow \boldsymbol{\mathit{x}}_{k-1}'$\;
	      \Break\;
		 }
		 \If{\textsf{satisfy}$(\boldsymbol{\mathit{x}}_i, \boldsymbol{\mathit{x}}'_{k-1})$}{
		   $\alpha \leftarrow (\mathit{t}_{k}-\mathit{t}_{k-1})/(\mathit{t}_{i}-\mathit{t}_{k-1})$\;
		   \For{$\ell\leftarrow 1$ to $D$ \do}{
		     $\mathit{x}_k^{\ell'}\leftarrow\alpha\cdot(\mathit{x}_{i}^{\ell}-\mathit{x}_{k-1}^{\ell'}) + \mathit{x}_{k-1}^{\ell'}$\;
	       }
	       \Break\;
	     }
	  }
	}
	
	\Return $\boldsymbol{\mathit{x'}}$\;
\end{algorithm}

\begin{example}[Local streaming, Example \ref{example:dirty} continued]
Consider again the time series $\boldsymbol{\mathit{x}}$ and the speed constraints $\mathit{s} = 1$ in Example \ref{example:dirty} with window size $\mathit{w} = 2$.

Let $\mathit{t}_k=2$, i.e. $\boldsymbol{\mathit{x}}_2=(1.8,1.8)$ be the key point, then the current window is \{$\boldsymbol{\mathit{x}}_2, \boldsymbol{\mathit{x}}_3, \boldsymbol{\mathit{x}}_4$\}.
According to Algorithm \ref{alg:local}, since $\boldsymbol{\mathit{x}}_2$ violates the speed constraint with $\boldsymbol{\mathit{x}}'_1=(1,1)$, we find the first succeeding point $i$ that \textsf{satisfy}$(\boldsymbol{\mathit{x}}_i, \boldsymbol{\mathit{x}}'_{1})$, i.e. $\boldsymbol{\mathit{x}}_3=(2.6,1)$.
A repair on $\boldsymbol{\mathit{x}}_2$ is :
\begin{align*} &
\mathit{x}^{1'}_2 = \frac{\mathit{t}_2-\mathit{t}_1}{\mathit{t}_3-\mathit{t}_1}  (2.6-1) + 1 = 1.8, & &
\mathit{x}^{2'}_2 = \frac{\mathit{t}_2-\mathit{t}_1}{\mathit{t}_3-\mathit{t}_1}  (1-1) + 1 = 1.
\end{align*}
When the window is \{$\boldsymbol{\mathit{x}}_5, \boldsymbol{\mathit{x}}_6, \boldsymbol{\mathit{x}}_7$\}, a repair on $\boldsymbol{\mathit{x}}_5$ is :
\begin{align*} &
\mathit{x}^{1'}_5 = \frac{\mathit{t}_5-\mathit{t}_4}{\mathit{t}_7-\mathit{t}_4}  (6.4-3.4) + 3.4 = 4.4, & &
\mathit{x}^{2'}_5 = \frac{\mathit{t}_5-\mathit{t}_4}{\mathit{t}_7-\mathit{t}_4}  (1-1) + 1 = 1.
\end{align*}
Similarly, a repair on $\boldsymbol{\mathit{x}}_6$ is performed, i.e. $\boldsymbol{\mathit{x}}_6'=(5.4,1)$.
Hence, the repair number is 3 and we have 
$\Delta(\boldsymbol{\mathit{x}}, \mathit{x}_\text{global}) = 2 \leq \Delta(\boldsymbol{\mathit{x}}, \mathit{x}_\text{local}) = 3$.
\end{example}

\subsection{Online Clustering}
\label{sect:cluster}

As mentioned before, \textsf{MTCSC-L} uses the first succeeding point $m$ with \textsf{satisfy}$(\boldsymbol{\mathit{x}}'_p, \boldsymbol{\mathit{x}}_{m})$ to help repairing.
However, the correctness of the point $m$ is not guaranteed, i.e. the potential outlier that satisfies the speed constraint can mislead the repair result.
Therefore, we propose \textsf{MTCSC-C} to fully utilize the data distribution in the current window.
We first show how to determine the candidate range.
Then we introduce the method \textsf{BuildCluster} to construct the cluster in the given window and explain the concept behind it.
Finally, we describe the whole procedure to obtain the final result.

\subsubsection{Candidate Range}
\label{sect:candidate-range}

Proposition \ref{the:pre-range} shows that we can always have a tighter range specified by the nearest preceding point.
In addition, the following proposition shows that this property also applies to succeeding points.

\begin{proposition}
\label{the:post-range}
For any $i > j > k, \mathit{t}_k<\mathit{t}_j<\mathit{t}_i<\mathit{t}_k+\mathit{w}$,  if \textsf{satisfy}$(\boldsymbol{\mathit{x}}_j, \boldsymbol{\mathit{x}}_i)$, 
then we have
$
\{\boldsymbol{\mathit{x}}'_k \mid d(\boldsymbol{\mathit{x}}'_k, \boldsymbol{\mathit{x}}_j)\leq\mathit{s}(\mathit{t}_j-\mathit{t}_k)\}
\subset
\{\boldsymbol{\mathit{x}}'_k \mid d(\boldsymbol{\mathit{x}}'_k, \boldsymbol{\mathit{x}}_i)\leq\mathit{s}(\mathit{t}_i-\mathit{t}_k)\}.
$
\end{proposition}

\begin{proof}
Referring to the proof of Proposition \ref{the:pre-range}, once 
$\boldsymbol{\mathit{x}}'_k$ is in the range specified by $\boldsymbol{\mathit{x}}_i$,  it is also in the range specified by $\boldsymbol{\mathit{x}}_j$.
\end{proof}

\begin{revised}
Hence, together with the assumption in Section \ref{sect:local-stream}, Proposition \ref{the:pre-range} and Proposition \ref{the:post-range} show that considering the nearest preceding and succeeding point (adjacent points) is sufficient to determine the final candidate range for the key point $k$. 
Since the preceding point $p$ is already fixed, we now turn to find a suitable succeeding point.
\end{revised}

\subsubsection{Cluster Construction}

\begin{algorithm}[!ht]
\caption{\textsf{BuildCluster}($\boldsymbol{\mathit{x}}'_p, \mathit{W}, \mathit{s}$)}
\label{alg:build-cluster}
  \KwIn{
    the last preceding point $\boldsymbol{\mathit{x}}'_p$, 
    the window $\mathit{W}$ starting from key point $\boldsymbol{\mathit{x}}_k$, 
    speed constraint $\mathit{s}$
  }
  \KwOut{cluster set $Clusters$ formed from $\mathit{W}$ except $\boldsymbol{\mathit{x}}_k$}
    $Map<i, Cluster> map \leftarrow \{\}$\;
    $f[1..\mathit{W}.length] \leftarrow \{0\}$\;
    \For{$\ell \leftarrow 1$ to $W.length$ \do}{
      \If{\textsf{satisfy}$(\boldsymbol{\mathit{x}}'_p, \mathit{W}[\ell])$}{
        $f[\ell]\leftarrow -1, map.put(\ell, new\ Cluster(\ell))$\;
        \Break\;
	  }
	}
	\For{$i\leftarrow\ell+1$ to $W.length$}{
      \For{$j \leftarrow i-1$ to $\ell$ \do}{
        \If{\textsf{satisfy}($\mathit{W}[i], \mathit{W}[j])$}{
          \If{$f[j] = -1$}{
            $f[i]\leftarrow j, Cluster(j).add(i)$\;
          } \ElseIf{$f[j] > 0$}{
            $f[i]\leftarrow f[j], Cluster(f[i]).add(i)$\;
          }
          \Break\;
        } \ElseIf{$j=\ell$ \textbf{or} $f[j]>0$} {
          \If{\textsf{satisfy}$(\boldsymbol{\mathit{x}}'_{p}, \mathit{W}[i])$}{
            $f[i] \leftarrow -1, map.put(i, new\ Cluster(i))$\;
          }
          \Break\;
        }
      } 
	} 
	\Return $map.values()$\;
\end{algorithm}

Algorithm \ref{alg:build-cluster} obtains the main trend of the data distribution in the window by clustering 
\begin{revised}
with a time complexity of $O(w^2D)$. 
\end{revised}
Line 2 creates an array that records the flag $f$ for each point: 
(1) $0$, the default value (dirty);
(2) $-1$, which is itself the first point of a cluster;
(3) a number greater than 0, the first point of the cluster to which it belongs.
Lines 3 to 6 find the first point $\ell$ that satisfies the speed constraint with the previous fixed point $p$ and create a new cluster containing it.
If no points are satisfied, the algorithm skips the for loop in Line 7 and returns an empty list.
Other points before $\ell$ are omitted as they are not in the range specified by the point $p$.
Starting from Line 7, for each point $i$ after $\ell$, we check its relationship with its preceding point $j = i - 1$ and do one of the four actions:
\begin{enumerate}[
leftmargin=*,
label={\textbf{Action} \arabic*},
]
\item Add $i$ into the cluster containing $j$ (Lines 9 to 14);
\item Create a new cluster for $i$ (Lines 15 to 18);
\item Continue to check before $j$ ($j=j-1$) with $i$ (Line 8);
\item Omit $i$ (Line 7).
\end{enumerate}

\paragraph{Explanations}
\begin{revised}
The reasons for Action 1 and Action 4 are obvious, i.e. we try to add `clean' points to existing clusters and omit the dirty points, respectively.
\end{revised}


In the following, we will discuss the reasons for other actions in detail, with respect to the three cases where \textsf{Satisfy}$(\mathit{W}[i], \mathit{W}[j=i-1])$ is not satisfied:
case 1: $i$ is dirty but $j$ is clean;
case 2: $i$ is clean but $j$ is dirty;
case 3: $i$ and $j$ are both dirty.

Action 2 creates a new cluster for the point $i$ if $j$ is not the first point of a cluster ($f[j]>0$), or all previous points in the window are incompatible, as long as $i$ is compatible with the fixed point $p$.
(1) In case 1, it guarantees that the dirty point $i$ forms a new cluster and does not collide with the existing clean cluster;
(2) In case 2, the clean point $i$ is not polluted by the existing dirty cluster;
(3) In case 3, dirty points with different properties are separated and the possibility of a dirty cluster becoming the largest is reduced.

Action 3 lets $i$ proceed to check the relation to $j-1$ if $j$ is the first point of a cluster ($f[j] = -1$) or $j$ has already been omitted ($f[j] = 0$).
(1) In cases 1 and 3, $i$ has three types of outcomes:
(a) becomes a member of a clean cluster, which, as mentioned above, has no effect on the final choice;
(b) becomes a member of a dirty cluster that corresponds to the purpose of collecting data with similar properties;
(c) creates a new cluster that is isolated from the others, as it is a dirty point.
So moving on is a smart option.
(2) In case 2, $i$ continues the search to join the clean cluster.

\subsubsection{Repair via Cluster}
\label{sect:repair-via-cluster}

\begin{algorithm}[!ht]
\caption{\textsf{MTCSC-C}($\boldsymbol{\mathit{x}}, \mathit{s}$)}
\label{alg:cluster}
  \KwIn{time series $\boldsymbol{\mathit{x}}$, speed constraint $\mathit{s}$}
  \KwOut{a repair $\boldsymbol{\mathit{x'}}$ of $\boldsymbol{\mathit{x}}$}
    $W \leftarrow \{\}$\;
    \For{$k \leftarrow 2$ to $n$ \do}{
      $W.add(\boldsymbol{\mathit{x}}_k)$\;
	  \For{$i \leftarrow k+1$ to $n$ \do}{
	    \If{$\mathit{t}_{i}>\mathit{t}_{k}+\mathit{w}$}{
	      \Break\;
        }
	    $W.add(\boldsymbol{\mathit{x}}_i)$\;
      } 
	  $Clusters \leftarrow \textsf{BuildCluster} (\boldsymbol{\mathit{x}}_{k-1}', W, \mathit{s})$\;
	  $i \leftarrow \argmax_j Cluster(j).size() + k$\;
      \If {
      \begin{revised}
      !(\textsf{satisfy}$(\boldsymbol{\mathit{x}}'_{k-1}, x_{k})$ \textbf{and} \textsf{satisfy}$(\boldsymbol{\mathit{x}}_k, \boldsymbol{\mathit{x}}_i)$)
      \end{revised}}{
	    $\alpha \leftarrow (\mathit{t}_{k}-\mathit{t}_{k-1})/(\mathit{t}_{i}-\mathit{t}_{k-1})$\;
	    \For{$\ell\leftarrow 1$ to $D$ \do}{
	      $\mathit{x}_k^{\ell'}\leftarrow\alpha\cdot(\mathit{x}_{i}^{\ell}-\mathit{x}_{k-1}^{\ell'}) + \mathit{x}_{k-1}^{\ell'}$\;
	    }
	  }
	  $W.clear()$\;
	} 
	
	\Return $\boldsymbol{\mathit{x}}'$\;
\end{algorithm}

Algorithm \ref{alg:cluster} introduces the online cleaning method via clustering.
For each point $k$ as the key point, Lines 3 to 7 collect all points in the window starting with $k$, and Line 8 constructs the clusters.
After collecting the data distribution of the given window, Line 9 considers the first point of the largest cluster as the representative for cleaning.
\begin{revised}
Line 10 recognizes whether $\boldsymbol{\mathit{x}}_{k}$ violates the speed constraint, so that the time complexity of the detection for each point is $O(w^2D+2D)$.
\end{revised}
In Lines 
\begin{revised}
11 to 13
\end{revised}
, $\boldsymbol{\mathit{x}}'_{k}$ is determined according to the formula (\ref{equation:local-stream}).
The time complexity of Algorithm \ref{alg:cluster} is 
\begin{revised}
$O(\mathit{w}^2Dn)$.
\end{revised}
As mentioned above, Algorithm \ref{alg:cluster} is an improvement of Algorithm \ref{alg:local} as it searches for a more ``accurate'' succeeding point to identify the candidate range.
Further discussion on the difference between them can be found in Section \ref{sect:experiment}.

\begin{example}[Repair via Cluster]
\label{example:repair-via-cluster}
\begin{figure}[t]
  \centering
  \includegraphics[width=\figwidths]{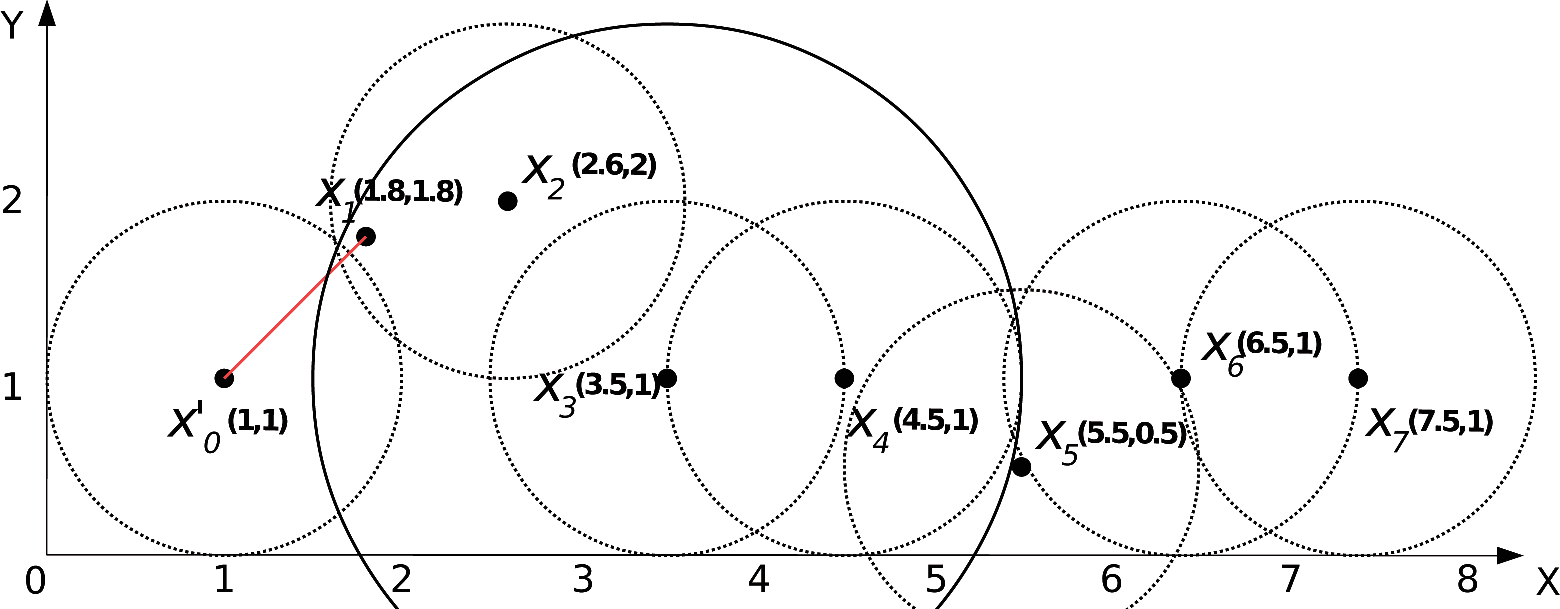}
  \caption{Build Cluster and Repair via Cluster}%
  \label{fig:cluster-example}
\end{figure}
Consider a time series ($D = 2$) $\boldsymbol{\mathit{x}} = \{(1,1), (1.8,1.8), (2.6,2), (3.5,1), (4.5,1), (5.5,0.5), (6.5,1), \allowbreak (7.5,1)\}$ of 8 points, with timestamps $\mathit{t} = \{0, 1, 2, 3, 4, 5, 6, 7\}$.
Figure \ref{fig:cluster-example} shows the data points (in black). 
We assume that $\boldsymbol{\mathit{x}}_0$ has been repaired.
Let the speed constraint $\mathit{s} = 1$ and windows size $\mathit{w} = 6$.

Let $\mathit{t}_k=1$, i.e. $\boldsymbol{\mathit{x}}_1$ be the key point, then the current window is \{$\boldsymbol{\mathit{x}}_1, \boldsymbol{\mathit{x}}_2, \ldots, \boldsymbol{\mathit{x}}_7$\} and the previous fixed point is $\boldsymbol{\mathit{x}}_0'$.
According to lines 3-6 of Algorithm \ref{alg:build-cluster}, $\boldsymbol{\mathit{x}}_2$ satisfies the speed constraint with $\boldsymbol{\mathit{x}}_0'$ ($\sqrt{(2.6-1)^2+(2-1)^2}/(2-0) \approx 0.94<1$).
So we create a new cluster containing $\boldsymbol{\mathit{x}}_2$.

Next, starting from line 7 of Algorithm \ref{alg:build-cluster}, we proceed to evaluate different scenarios for the following points. We only show the first two points after $\boldsymbol{\mathit{x}}_2$ due to the space limit.

$\boldsymbol{\mathit{x}}_3$ : 
$\boldsymbol{\mathit{x}}_3$ is not compatible with $\boldsymbol{\mathit{x}}_2$ and there are no data points before $\boldsymbol{\mathit{x}}_2$, but it is compatible with $\boldsymbol{\mathit{x}}_0'$, thus we create a new cluster for $\boldsymbol{\mathit{x}}_3$ (Action 2 and $j=\ell$).

$\boldsymbol{\mathit{x}}_4$ : $\boldsymbol{\mathit{x}}_4$ is compatible with $\boldsymbol{\mathit{x}}_3$, we add $\boldsymbol{\mathit{x}}_4$ into the cluster containing $\boldsymbol{\mathit{x}}_3$ (Action 1).




In the end, three clusters will be formed, namely \{$\boldsymbol{\mathit{x}}_2$\}, \{$\boldsymbol{\mathit{x}}_3$, $\boldsymbol{\mathit{x}}_4$, $\boldsymbol{\mathit{x}}_6$, $\boldsymbol{\mathit{x}}_7$\} and \{$\boldsymbol{\mathit{x}}_5$\}.
According to Algorithm \ref{alg:cluster}, the repair based on $\boldsymbol{\mathit{x}}_0'$ and $\boldsymbol{\mathit{x}}_3$ for $\boldsymbol{\mathit{x}}_1$ is:
\begin{align*} &
\mathit{x}^{1'}_1 = \frac{\mathit{t}_1-\mathit{t}_0}{\mathit{t}_3-\mathit{t}_0}  (3.5-1) + 1 \approx 1.83, & &
\mathit{x}^{2'}_1 = \frac{\mathit{t}_1-\mathit{t}_0}{\mathit{t}_3-\mathit{t}_0}  (1-1) + 1 = 1.
\end{align*}
The final repair result obtained through Algorithm \ref{alg:build-cluster} and Algorithm \ref{alg:cluster} is $\boldsymbol{\mathit{x}}_1' = (1.83,1), \boldsymbol{\mathit{x}}_2' = (2.66,1)$ and $\boldsymbol{\mathit{x}}_5' = (5.5,1)$.
\end{example}


\section{Adaptive Speed Constraint}
\label{sect:adaptive}

A precise speed constraint is crucial for the performance of all proposed methods.
\cite{DBLP:conf/sigmod/SongZWY15, DBLP:journals/tods/SongGZWY21} capture the speed constraint via 
(1) domain experts or common sense;
(2) extraction from the data by the $95\%$ confidence level \cite{jerrold1999biostatistical}.
%
However, it is still not trivial to set an appropriate speed constraint, especially in the multivariate case, since
(1) we may not obtain sufficient domain knowledge, especially if the dimensions are only weakly correlated;
(2) the confidence level $95\%$ is not always large enough to set an appropriate speed constraint;
(3) the time series (stream) is often non-stationary \cite{DBLP:journals/jips/DingZLWWL23}, i.e. the properties of the data may change.
In such cases, the predefined speed constraint may take effect in that hour, but may no longer bind the data correctly in the next hour.
For example, a person may change their modes of transportation in a short time \cite{DBLP:conf/www/ZhengLWX08}, e.g. walk to the bus stop and then take the bus. 

To solve the above problems, we develop an adaptive method to dynamically set an appropriate speed constraint by monitoring the difference in data distribution between adjacent sliding windows.

\subsection{Speed Distribution}
Theoretically, we should collect a large number of data points to model different phases of an evolving time series.
However, it is not possible to obtain such training data for each dataset, as different people and spaces have their own characteristics, e.g. different walking speeds.
We therefore intend to use the observations directly to capture the speed distribution.

Another problem arises from the sparsity of time series whose data are continuous.
If we want the frequency distribution to work well, we need to collect our data in buckets to get a high enough probability for each category \cite{DBLP:journals/vldb/WangZSW24}.
Then we use the KL divergence \cite{kullback1997information} to measure the distance between two distributions \cite{DBLP:journals/pvldb/DasuL12}.

\subsection{Adaptive Speed}
\label{sect:adaptive-speed}

\begin{figure}[t]
\centering
\begin{minipage}{\figwidths}
\hspace{-0.5em}%
\includegraphics[width=\figwidths]{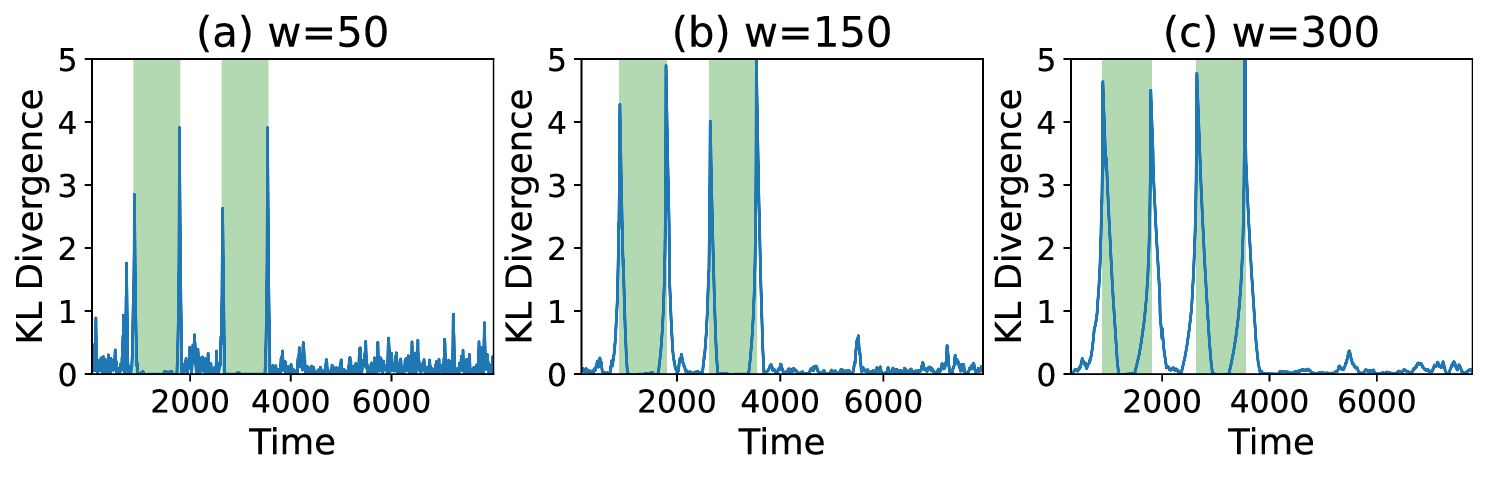}%
\end{minipage}
\begin{minipage}{\figwidths}
\hspace{-0.5em}%
\includegraphics[width=\figwidths]{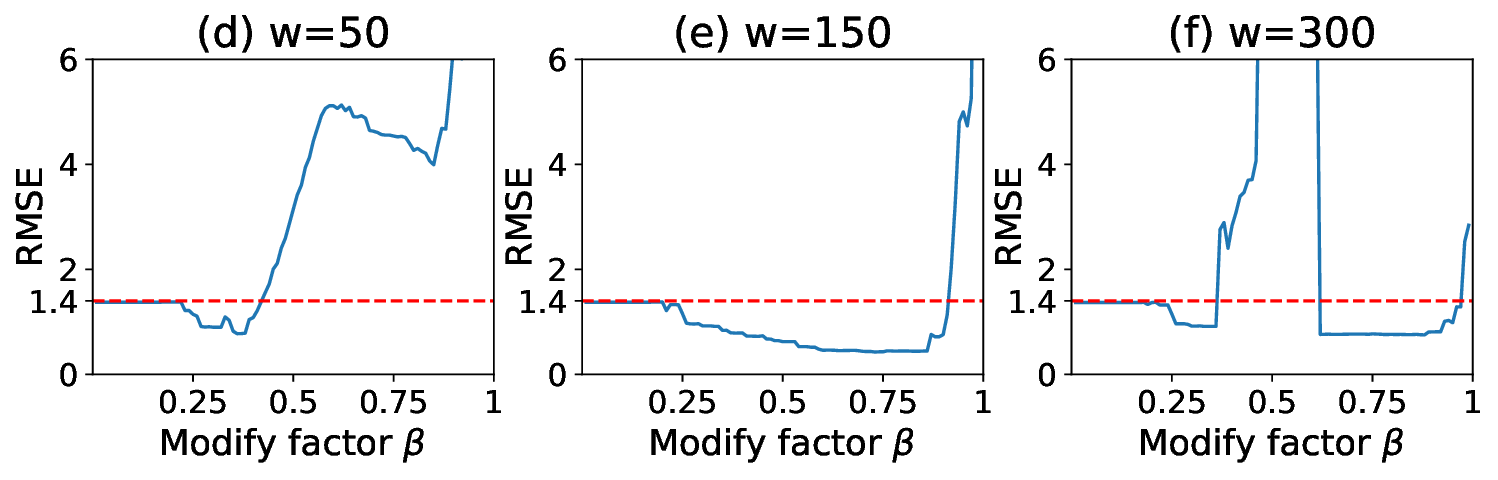}%
\end{minipage}
\caption{KL divergence/RMSE under different window sizes}
\label{fig:KL}
\end{figure}

\paragraph{Solution}
\begin{revised}
We propose an online cleaning method \textsf{MTCSC-A} with adaptive speed.
Algorithm \ref{alg:adaptive} is the core part and introduces how to set the speed constraint dynamically with a complexity of $O(D+b+m)$.
\end{revised}
First, we construct the initial speed distribution in the first window $\mathit{W}_1$ and the second window $\mathit{W}_2$, where $\mathit{W}_1$ and $\mathit{W}_2$ are adjacent and disjoint.
\begin{revised}
$\mathit{W}_1$ and $\mathit{W}_2$ store the speed instead of the values of data points to save computation time.
\end{revised}
When data points arrive, we update the distribution in $\mathit{W}_1, \mathit{W}_2$ and calculate the distance between them.
Once this is greater than a threshold $\tau$, we change the speed constraint with a modify factor $\beta$ ($\mathit{s}' = \mathit{s}_{95\%}(in\ \mathit{W}_2)/\beta)$ and also update the distributions in $\mathit{W}_1, \mathit{W}_2$ (since the points in the buckets also change).
\begin{revised}
\textsf{UpdateDistribution}$(\mathit{W}_1, \mathit{b}, \mathit{s})$ represents dividing the speed stored in $\mathit{W}_1$ into $\mathit{b}$ buckets at equal intervals based on the current speed constraint $\mathit{s}$.
\textsf{MTCSC-A} is formed by inserting Algorithm \ref{alg:adaptive} through ``$\mathit{s}$ = \textsf{AdaptiveSpeed}($\boldsymbol{\mathit{x}}_{k-1}, \boldsymbol{\mathit{x}}_k, \mathit{s}$, $\mathit{b}$, $\tau$, $\mathit{m}$, $\beta$, $\mathit{W}_1$, $\mathit{W}_2$)'' between the 3rd and 4th lines of Algorithm \ref{alg:cluster}.
Overall, the time complexity of \textsf{MTCSC-A} is $O((\mathit{D}+\mathit{b}+\mathit{m}+\mathit{w}^{2}D+2D)*n)=O(\mathit{w}^{2}Dn)$, the same with \textsf{MTCSC-C}.
\end{revised}

\begin{algorithm}[!ht]
\begin{revised}
\caption{\textsf{AdaptiveSpeed}($\boldsymbol{\mathit{x}}_p, \boldsymbol{\mathit{x}}_k, \mathit{s}$, $\mathit{b}$, $\tau$, $\mathit{m}$, $\beta$, $\mathit{W}_1$, $\mathit{W}_2$)}
\label{alg:adaptive}
  \KwIn{the last preceding point $\boldsymbol{\mathit{x}}_p$, the key point $\boldsymbol{\mathit{x}}_k$, speed constraint $\mathit{s}$, bucket number $\mathit{b}$, threshold $\tau$, monitoring interval $\mathit{m}$, modify factor $\beta$, the first window $\mathit{W}_1$, the second window $\mathit{W}_2$}
  \KwOut{speed constraint $\mathit{s}'$}
	$\mathit{s}' = \mathit{s}$, $\mathit{s}_1 \leftarrow d(\boldsymbol{\mathit{x}}_k, \boldsymbol{\mathit{x}}_p)/(\mathit{t}_k-\mathit{t}_p)$\;
    \If{$\mathit{W}_1.size < \mathit{m}$}{
      $\mathit{W}_1.push(\mathit{s}_1)$, \textsf{UpdateDistribution}$(\mathit{W}_1, \mathit{b}, \mathit{s})$\;
    } \ElseIf{$\mathit{W}_2.size < \mathit{m}$}{
      $\mathit{W}_2.push(\mathit{s}_1)$, \textsf{UpdateDistribution}$(\mathit{W}_2, \mathit{b}, \mathit{s})$\;
    } \ElseIf{$\mathit{W}_2.size = \mathit{m}$}{
	  \If{$KL(\mathit{W}_1, \mathit{W}_2) > \tau$}{
		$\mathit{s}' = \mathit{s}_{95\%}(in\ \mathit{W}_2)/\beta$\;
	  }	
	  $\mathit{s}_2 \leftarrow \mathit{W}_2.front()$, $\mathit{W}_2.pop()$\;
	  $\mathit{W}_1.push(\mathit{s}_2)$, $\mathit{W}_1.pop()$, \textsf{UpdateDistribution}$(\mathit{W}_1, \mathit{b}, \mathit{s}')$\;
	  $\mathit{W}_2.push(\mathit{s}_1)$, \textsf{UpdateDistribution}$(\mathit{W}_2, \mathit{b}, \mathit{s}')$\;
	}
	
	\Return $\mathit{s}'$\;
\end{revised}
\end{algorithm}

\paragraph{Settings}
We have the following four hyper-parameters.
\begin{itemize}[fullwidth]
\item
Bucket number $b$. 
Figures \ref{exp:sensitivity}(a) show that $b$ it is not sensitive.
\item
Threshold $\tau$.
As can be seen in Figures \ref{fig:KL}(a, b, c), there is a clear difference between the KL divergence at the time when the speed changes or not (green blocks indicate a different mode), regardless of the window size.
\item
\begin{revised}
Monitoring interval $\mathit{m}$.
\end{revised}
Obviously there is a trade-off of 
\begin{revised}
$\mathit{m}$
\end{revised}
, as shown in Figure \ref{fig:KL}.
(1) For a small 
\begin{revised}
$\mathit{m} = 50$
\end{revised}
, it is difficult to determine the updated $\mathit{s}$ because the speed distribution may not be reliable due to the small number of data points.
The red line in Figure \ref{fig:KL}(d) shows the RMSE of the observations without any repair.
We find that only a small interval of $\beta$ can lead to better performance.
(2) On the other hand, 
\begin{revised}
$\mathit{m}=300$
\end{revised}
 is large enough to collect a number of data
points that $\mathit{s}_{95\%}$ is relatively accurate.
So if $\beta$ is small, the updated constraint $\mathit{s}$ is too large to identify the error points.
At the same time, the updated constraint must not be too small, as this would mislead the cluster result.
However, such a large 
\begin{revised}
$\mathit{m}$
\end{revised}
also reduces the range in which $\tau$ can be selected. 
We find that $2$ is sufficient for 
\begin{revised}
$\mathit{m} = 50$
\end{revised}
, $3$ for 
\begin{revised}
$\mathit{m} = 150$
\end{revised}
and $4$ or more for 
\begin{revised}
$\mathit{m} = 300$
\end{revised}
.
(3) Intuitively, we would consider a ``proper'' setting with a moderate value, for example 
\begin{revised}
$\mathit{m} = 150$
\end{revised}
 in the figure.
Following this guide, a good 
\begin{revised}
$\mathit{m}$
\end{revised}
 could practically be chosen by observing the width of the KL divergence.
\item
Modify factor $\beta$.
It is related to the setting of 
\begin{revised}
$\mathit{m}$
\end{revised}
 and has already been discussed above.
\end{itemize}

\begin{example}[Adaptive speed]
Let $\mathit{s}=2.2$ and $b=6$, the speed bucket is \{[0,0.44], (0.44,0.88], (0.88,1.32], (1.32,1.76], (1.76,2.2], (2.2, $+\infty$]\}.
The last bucket is placed at a point that exceeds the current speed constraint.
Consider the speed distribution of the first window $\mathit{W}_1=\{0, 0, 0, 90, 60, 0\}$ and second window $\mathit{W}_2=\{3, 4, 1, 44, 25,\allowbreak 73\}$ with $\tau = 0.75$, 
\begin{revised}
$\mathit{m} = 150$
\end{revised}
 and $\beta = 0.75$. 

For $\mathit{W}_1$ and $\mathit{W}_2$, the probability distribution 
$\mathit{P}_1  = \{0, 0, 0, 0.6, 0.4, 0\}$ and 
$\mathit{P}_2  = \{0.02, 0.027, 0.007, 0.293, 0.167, 0.487\}$.
The calculation process of KL divergence between distributions $\mathit{P}_1$ and $\mathit{P}_2$ is:
\begin{align*}
D_{KL}(\mathit{P}_1||\mathit{P}_2) &= 0 \times \log \frac{0}{0.02} + 0 \times \log \frac{0}{0.027} + 0 \times \log \frac{0}{0.007} \\
&\quad + 0.6 \times \log \frac{0.6}{0.293} + 0.4 \times \log \frac{0.4}{0.167} + 0 \times \log \frac{0}{0.487} \\
&= 0 + 0 + 0 + 0.6 \times (0.717) + 0.4 \times (0.873) + 0 \\
&= 0.7794
\end{align*}
Since the $\mathit{D}_{KL}$ is greater than the threshold $\tau$, $\mathit{s}$ is thus  modified to $\mathit{s}_{95\%}(in\ \mathit{W}_2)/\beta = 3.572 / 0.75 \approx 4.763$.
Meanwhile, the speed distribution of $\mathit{W}_1$ and $\mathit{W}_2$ will be updated using the latest speed bucket (\{[0,0.9526], (0.9526,1.9052], (1.9052,2.8578], (2.8578,3.8104], (3.8104,4.763], (4.763, $+\infty$]\}), with $\mathit{W}_1 = \{0, 131, 19, 0, 0, 0\}$ and $\mathit{W}_2 = \{7, 62, 14, 61, 4, 2\}$.
There are only two points that exceed the latest speed, indicating that our latest speed settings are relatively accurate.
\end{example}


\section{Experiment}
\label{sect:experiment}

In this section, we use experiments to evaluate performance
(1) on univariate time series with injected errors;
(2) on multivariate time series with different degrees of correlation and injected errors;
(3) on manually collected GPS trajectories (one of them with mode changes) with real errors;
and (4) the improvement of data mining tasks such as classification and clustering.
Finally, we summarize the highlights and discuss our limitations.
The code and data of this work are available online\footnote{\url{https://anonymous.4open.science/r/mtcsc-E4CC}}.

\subsection{Settings}
We run experiments on a Windows 10 laptop with a 2.30GHz Intel Core CPU and 32GB RAM.
Deep-learning methods are employed with GPU without comparing the efficiency.

\subsubsection{Datasets}
\label{sect:datasets}
We use seven public real-world datasets that are originally clean or clean after simple pre-processing with various sizes and dimensions.
To analyze the performance, we generate synthetic errors following the guideline in \cite{DBLP:conf/sigmod/SongZWY15}.
\begin{revised}
Errors are injected by randomly replacing the values of some data points in a given dimension.
For each replaced data point, it takes a random value between the minimum and maximum values in the dataset.
\end{revised}
Moreover, we collect two sets of GPS data via smartphones. 
One contains only walking data, the other is a mixture of walking, running and cycling.
The embedded errors are manually labeled according to the exact routes and maps.
In order to avoid the effect of randomness, we run experiments on datasets with synthetic errors 10 times with different generating seeds and report the average results.
The summary of all datasets is shown in Table \ref{table:dataset}.

\begin{table}[t]
 \caption{Summary of datasets}
 \label{table:dataset}
 \centering
 \resizebox{0.8\expwidths}{!}{%
 \begin{tabular}{ccccc}
 \toprule
 & Size & \#Dim & Error & \#Series \\
 \midrule
 Stock \cite{DBLP:conf/sigmod/SongZWY15} & $12$k & $1$ & Clean & $1$ \\
 ILD \cite{ILDdata} & $43$k & $3$ & Clean after pre-process & $1$ \\
 Tao \cite{DBLP:conf/cikm/AngiulliF07a} & $568$k & $3$ & Clean after pre-process & $1$ \\
 ECG\cite{yoon2020ultrafast} & $94$k & $32$ & Clean after pre-process & $1$ \\
 \midrule
 GPS(Walk) & $11$k & $2$ & Embedded & $1$ \\
 GPS(Mixed) & $8$k & $2$ & Embedded & $1$ \\
 \midrule
 ArrowHead \cite{UCRArchive} & $251$ & $1$ & Clean & $211$ \\
 \begin{revised}
 AtrialFib \cite{DBLP:journals/corr/abs-1811-00075} \end{revised} & $640$ & $2$ & Clean & $30$ \\
 DSR \cite{UCRArchive} & $345$ & $1$ & Clean & $16$ \\
 \begin{revised}
 SWJ \cite{DBLP:journals/corr/abs-1811-00075}\end{revised} & $2500$ & $4$ & Clean & $27$ \\
 \bottomrule
 \end{tabular}
}
\end{table}

\subsubsection{Metrics}
\begin{itemize}[fullwidth]
\item 
RMSE \cite{DBLP:conf/vldb/JefferyGF06} is employed to evaluate the repair accuracy.
\item
Repair distance $\delta(\boldsymbol{\mathit{x}}',\boldsymbol{\mathit{x}}) = \sum_{i=1}^n d(\boldsymbol{\mathit{x}}'_i, \boldsymbol{\mathit{x}}_i)/n$ presents the total distance between the time series before and after repair.
\item
Repair number $\Delta(\boldsymbol{\mathit{x}}', \boldsymbol{\mathit{x}}) = \sum_{i=1}^{n}\mathbb{I}(\boldsymbol{\mathit{x}}_i' \neq \boldsymbol{\mathit{x}}_i)/n$ shows the extent to which a method changes the input data.
\item
Time cost. Apart from loading data, error injection and result evaluation, the total time of other procedures is reported.
\end{itemize}

\subsubsection{Baselines}
\label{sect:baseline}
We compare our proposals with 9 competing baselines of different kinds in Table \ref{table:baseline}.
\begin{revised}
\textsf{SCREEN} and \textsf{SpeedAcc} use the univariate speed constraint and follow the widely used minimum change principle.
They are selected to verify the validity of the proposed minimum fix principle.
\textsf{HTD} calculates the correlation between the dimensions and performs data cleaning over multivariate time series based on speed constraint.
\textsf{RCSWS} aims to clean GPS data and can only implement on two-dimensional data.
\textsf{LsGreedy} is a statistical SOTA method that has good accuracy.
Due to the inspiration from interpolation in choosing repair solution, we also choose the classical smoothing-based \textsf{EWMA} method as a baseline.
\textsf{Holoclean} is a representative of the machine learning method, but it is proposed for relational data.
We thus convert time series to relational data, treating timestamps and the names of each dimension as attributes and the values of each dimension as cells.
We also convert the speed constraint of each dimension to the form of a denial constraint (which can be recognized by \textsf{HoloClean}) and add the parser module into the source code.
\end{revised}
We find that there are few works that use deep leaning techniques for time series cleaning. 
Therefore, we select two SOTA anomaly detection methods, \textsf{TranAD} and \textsf{CAE-M}, and consider their predicted/reconstructed values as repair candidates.
We recommend \textsf{MTCSC-C} as the general method and refer to it as \textsf{MTCSC} in the following experiments.

\begin{table}[t]
 \caption{Summary of compared methods}
 \label{table:baseline}
 \centering
 \resizebox{0.8\expwidths}{!}{%
 \begin{tabular}{cccc}
 \toprule
 Algorithm & Dimension & Process & Type \\
 \midrule
 \textsf{MTCSC-G} & multivariate & batch & constraint \\
 \textsf{MTCSC-L} & multivariate & online & constraint \\
 \textsf{MTCSC-C} & multivariate & online & constraint + statistical \\
 \textsf{MTCSC-A} & multivariate & online & constraint + statistical \\
 \midrule
 \textsf{SCREEN} \cite{DBLP:conf/sigmod/SongZWY15} & univariate & online & constraint \\
 \textsf{SpeedAcc} \cite{DBLP:journals/tods/SongGZWY21} & univariate & online & constraint \\
 \textsf{LsGreedy} \cite{DBLP:conf/sigmod/ZhangSW16} & univariate & online & statistical \\
 \textsf{EWMA} \cite{GARDNER2006637} & univariate & online & smoothing \\
 \textsf{RCSWS}\cite{DBLP:journals/tist/FangWYX22} & multivariate & online & constraint + statistical \\
 \textsf{HTD} \cite{DBLP:journals/ahswn/ZhouYZSMY22} & multivariate & batch & constraint \\
 \textsf{HoloClean} \cite{DBLP:journals/pvldb/RekatsinasCIR17} & multivariate & batch & machine learning \\
 \textsf{TranAD}\cite{DBLP:journals/pvldb/TuliCJ22} & multivariate & online & deep learning \\
 \textsf{CAE-M}\cite{DBLP:journals/corr/abs-2107-12626} & multivariate & online & deep learning \\
 \bottomrule
 \end{tabular}
}
\end{table}

\subsection{Comparison on Univariate Data}
Although we focus on the multivariate case, we also investigate whether we can outperform the univariate specific competitors on univariate time series.

\subsubsection{Comparison among our proposals}

\begin{figure}[t]
  \centering
  \includegraphics[width=\expwidths]{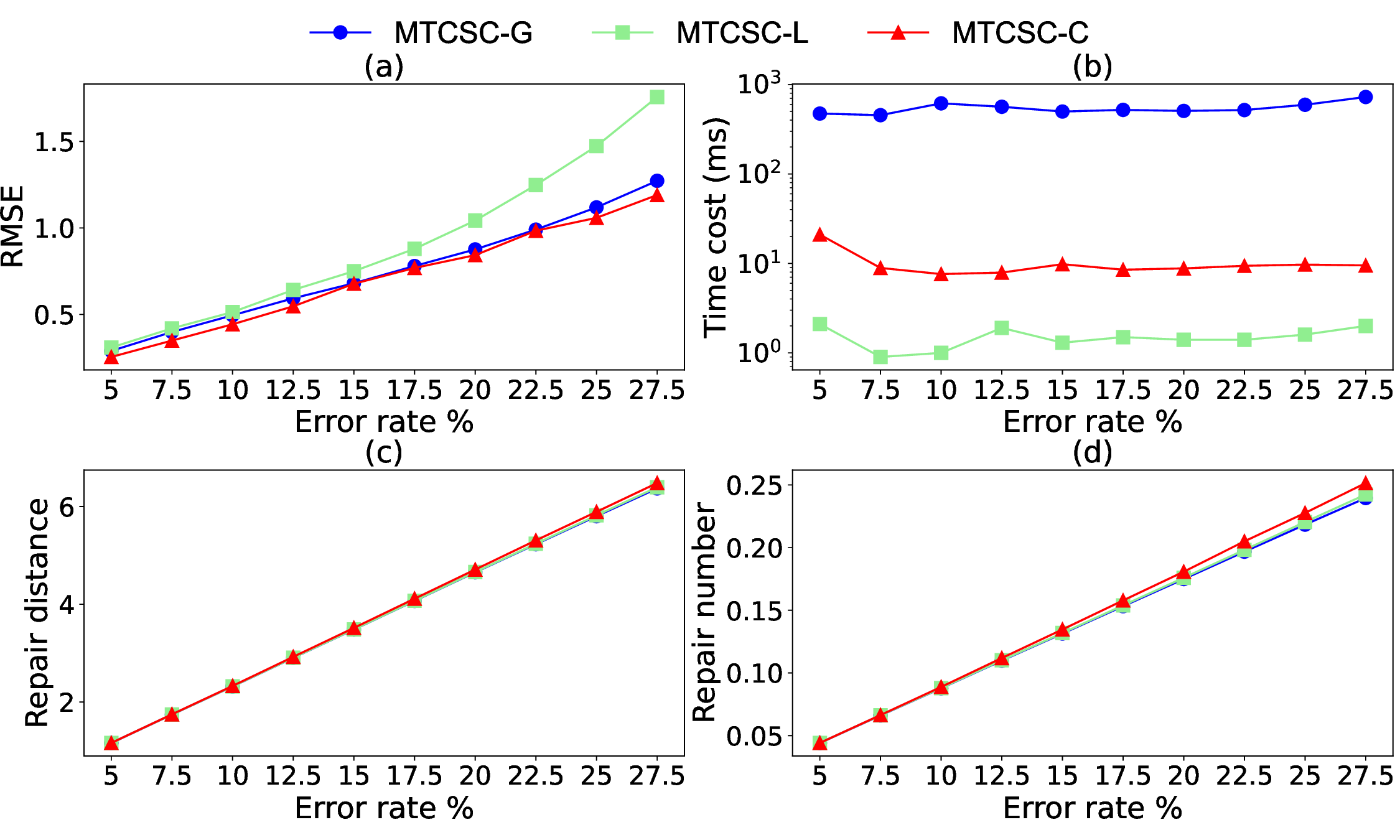}
  \caption{Varying error rate on Stock}%
  \label{exp:stock}
\end{figure}

We compare our proposals on Stock and the results are shown in Figure \ref{exp:stock}, which verifies our theory in Section \ref{sect:stream}.
Figure \ref{exp:stock}(a) shows that \textsf{MTCSC-G} and \textsf{MTCSC-C} behave similarly and are better than \textsf{MTCSC-L}.
Figure \ref{exp:stock}(b) shows that \textsf{MTCSC-G} takes the most time, while \textsf{MTCSC-L} takes the least time.
\textsf{MTCSC-C} achieves better results than \textsf{MTCSC-G} and takes less than half the time.
Figure \ref{exp:stock}(d) shows that \textsf{MTCSC-G} modifies the fewest points, while both \textsf{MTCSC-L} and \textsf{MTCSC-C} modify slightly more points compared to \textsf{MTCSC-G}.

\subsubsection{Varying Error Rate $e\%$}
\label{sect:uni-drate}

\begin{figure}[t]
  \centering
  \includegraphics[width=\expwidths]{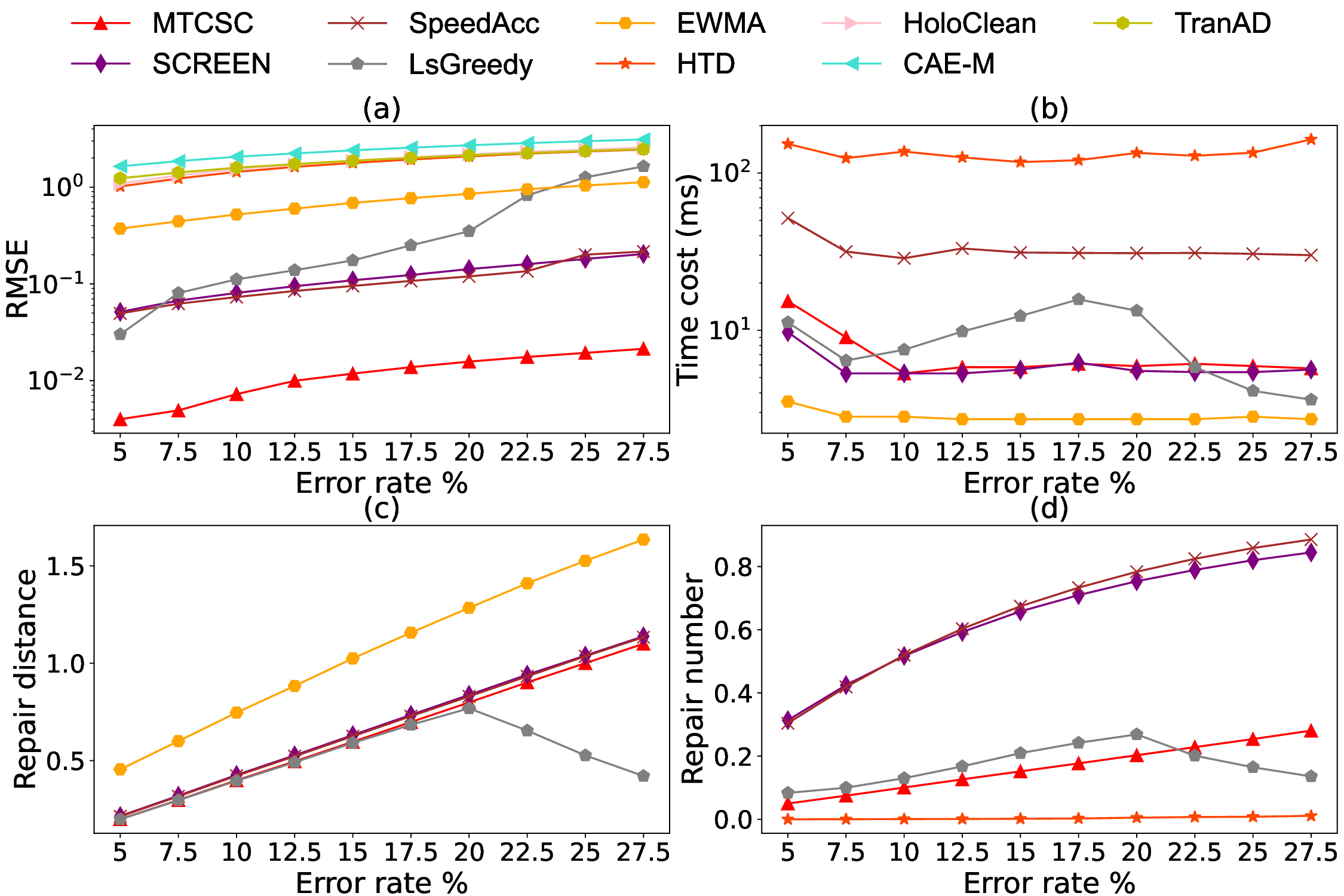}
  \caption{Varying error rate on Temperature (ILD)}%
  \label{exp:temp-drate}
\end{figure}


Figure \ref{exp:temp-drate}(a) shows that the proposed \textsf{MTCSC} achieves the best accuracy and is robust to error rates.
When the error rate reaches 20\%, the RMSE of \textsf{LsGreedy} suddenly jumps and shows a strong growth trend.
The direct reason for this can be found in Figure \ref{exp:temp-drate}(d), that \textsf{LsGreedy} repairs fewer points then.
The deeper insight is that the statistical based method \textsf{LsGreedy} treats the wrong points as correct ones when the error rate is high.
Figure \ref{exp:temp-drate}(b) shows that \textsf{MTCSC} has good efficiency.
\begin{revised}
As shown in Figures \ref{exp:temp-drate}(c) and (d),
we find that \textsf{SCREEN}, \textsf{SpeedAcc}, \textsf{LsGreedy}, and \textsf{HTD},
which are based on the minimum change principle, modify many more points than our proposed \textsf{MTCSC}, which is based on the minimum fix principle, resulting in an even larger repair distance.
Moreover, the number of points modified by \textsf{MTCSC} is relatively similar to the number of manually injected errors.
\end{revised}
Considering the accuracy shown in Figure \ref{exp:temp-drate}(a), the use of \emph{minimum fix principle} makes sense.

\subsubsection{Varying Data Size $n$}

\begin{figure}[t]
  \centering
  \includegraphics[width=\expwidths]{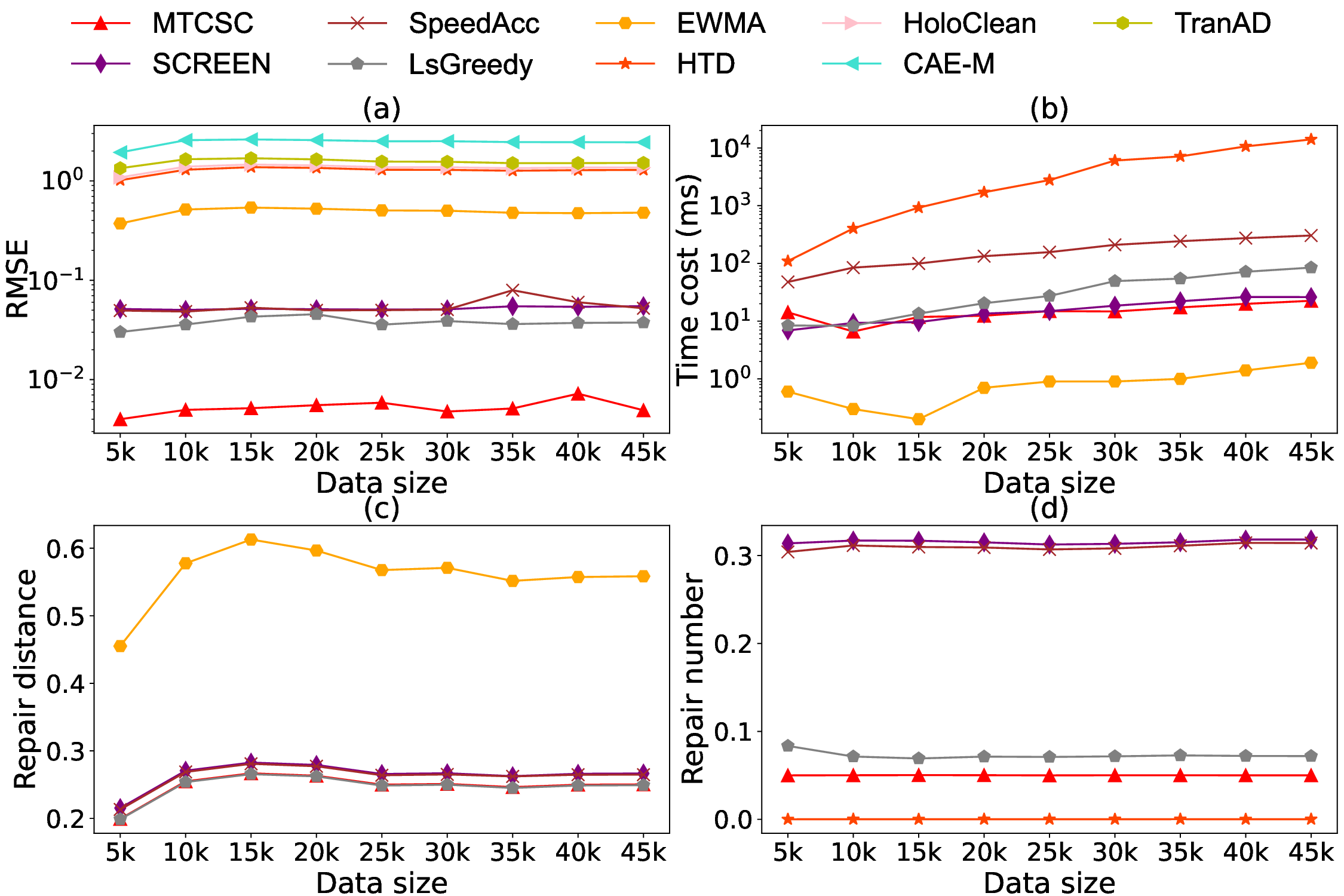}
  \caption{Varying data size on Temperature (ILD)}%
  \label{exp:temp-size}
\end{figure}


Figures \ref{exp:temp-size}(a) and (b) show that \textsf{MTCSC} also has good scalability in terms of data volume, achieves high accuracy and has a similar time cost as most fast methods with the exception of the \textsf{EWMA}.
Figure \ref{exp:temp-size}(d) shows that \textsf{MTCSC} repairs about 5\% of the points in each data size, which corresponds to the number of injected errors and indicates that \textsf{MTCSC} mainly cleans dirty points.
It can be seen that both the machine learning-based \textsf{HoloClean} and the deep learning-based \textsf{TranAD} and \textsf{CAE-M} perform relatively poorly on univariate data.

\subsection{Comparison on Multivariate Data}
\label{sect:multivariate}

We perform experiments on multivariate time series with different error rates, data sizes and error patterns.
Like other univariate methods, we also evaluate \textsf{MTCSC-Uni} by applying \textsf{MTCSC} in each dimension separately to analyze whether considering the entire dimensions as a whole has an effect.

\subsubsection{Varying Error Rate $e\%$}

\begin{figure}[t]
  \centering
  \includegraphics[width=\expwidths]{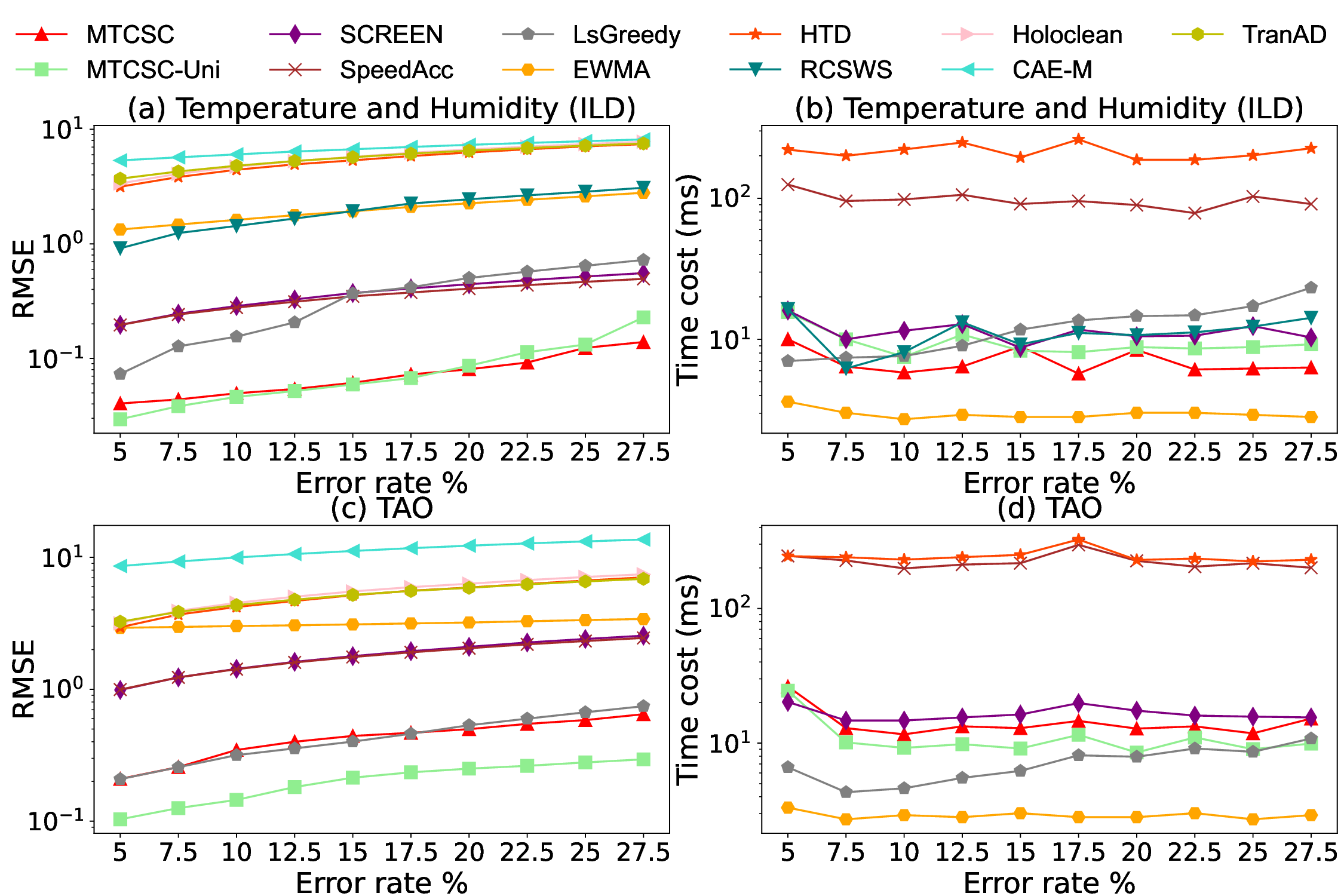}
  \caption{Varying error rate on (a-b) ILD (c-d) TAO (separate)}%
  \label{exp:drate-separate}
\end{figure}

\begin{figure}[t]
  \centering
  \includegraphics[width=\expwidths]{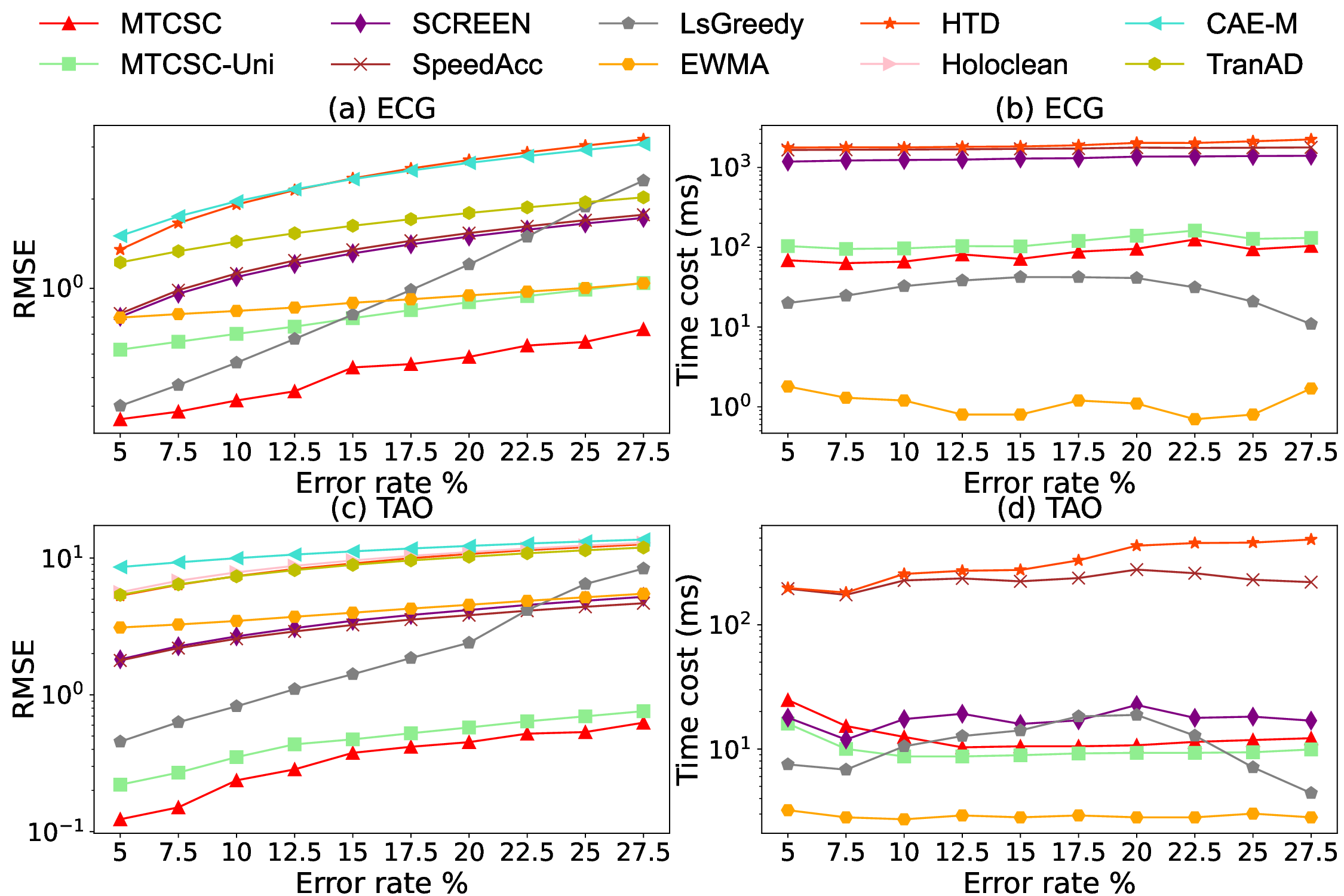}
  \caption{Varying error rate on (a-b) ECG (c-d) TAO (together)}%
  \label{exp:drate-together}
\end{figure}


Figure \ref{exp:drate-separate}(a) shows that our \textsf{MTCSC} and \textsf{MTCSC-Uni} achieve higher accuracy on ILD.
Figures \ref{exp:drate-together}(a) and (b) show that \textsf{MTCSC} also has better accuracy with significantly less time compared to \textsf{SCREEN} and \textsf{SpeedAcc} on high-dimensional (32) ECG data.
In Figure \ref{exp:drate-separate}(c), the performance of \textsf{LsGreedy} on separate patterns on TAO is comparable to that of \textsf{MTCSC}, but \textsf{MTCSC-Uni} outperforms \textsf{LsGreedy} by far.
In Figure \ref{exp:drate-together}(c), \textsf{MTCSC} and \textsf{MTCSC-Uni} lead in performance on TAO, with \textsf{MTCSC} slightly ahead of \textsf{MTCSC-Uni}.
Figures \ref{exp:drate-separate}(b) and (d), Figures\ref{exp:drate-together}(b) and (d) show that \textsf{MTCSC} and \textsf{MTCSC-Uni} do not lose efficiency when performance is increased.

\subsubsection{Varying Data Size $n$}

\begin{figure}[t]
  \centering
  \includegraphics[width=\expwidths]{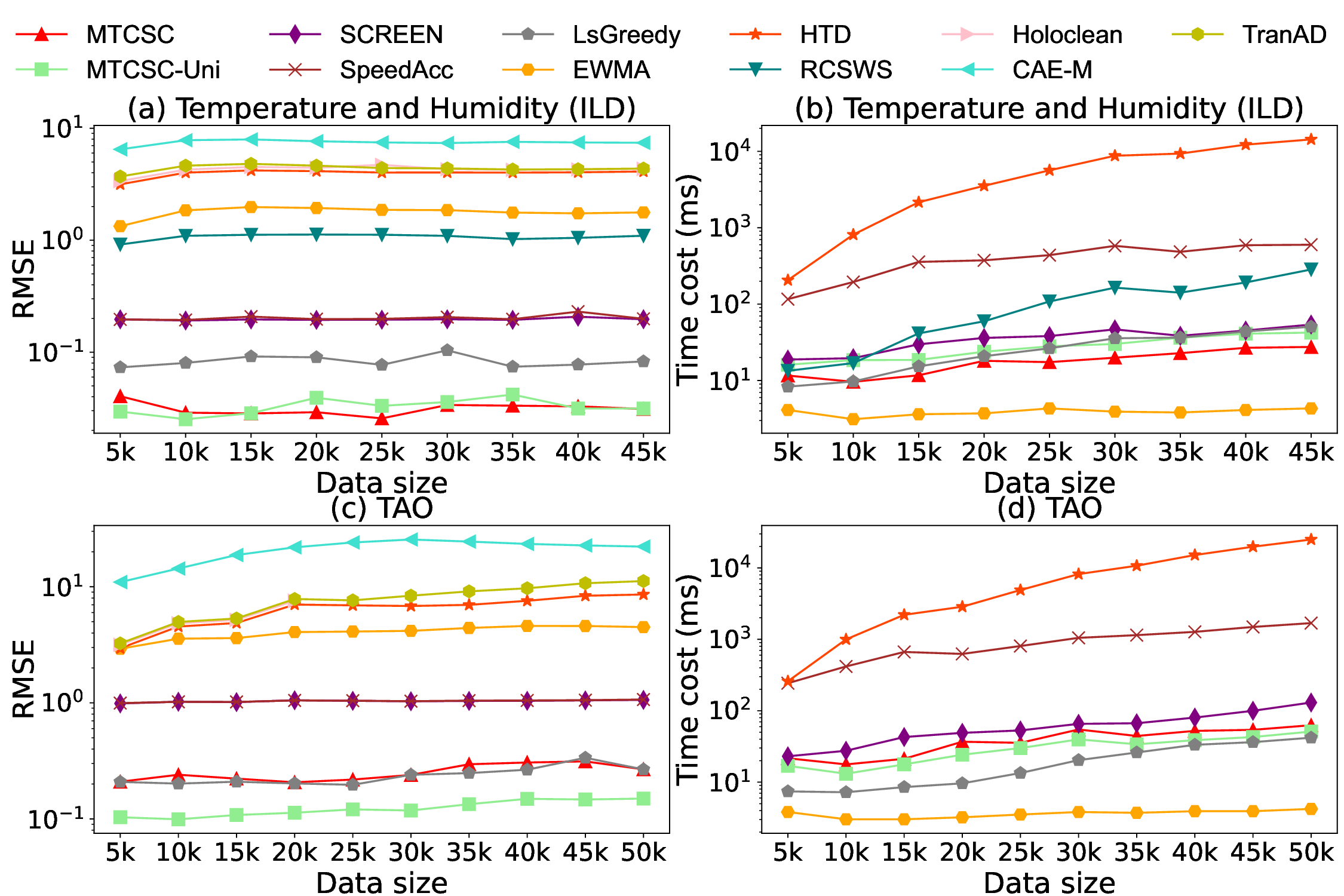}
  \caption{Varying data size on (a-b) ILD (c-d) TAO (separate)}%
  \label{exp:size-separate}
\end{figure}

\begin{figure}[t]
  \centering
  \includegraphics[width=\expwidths]{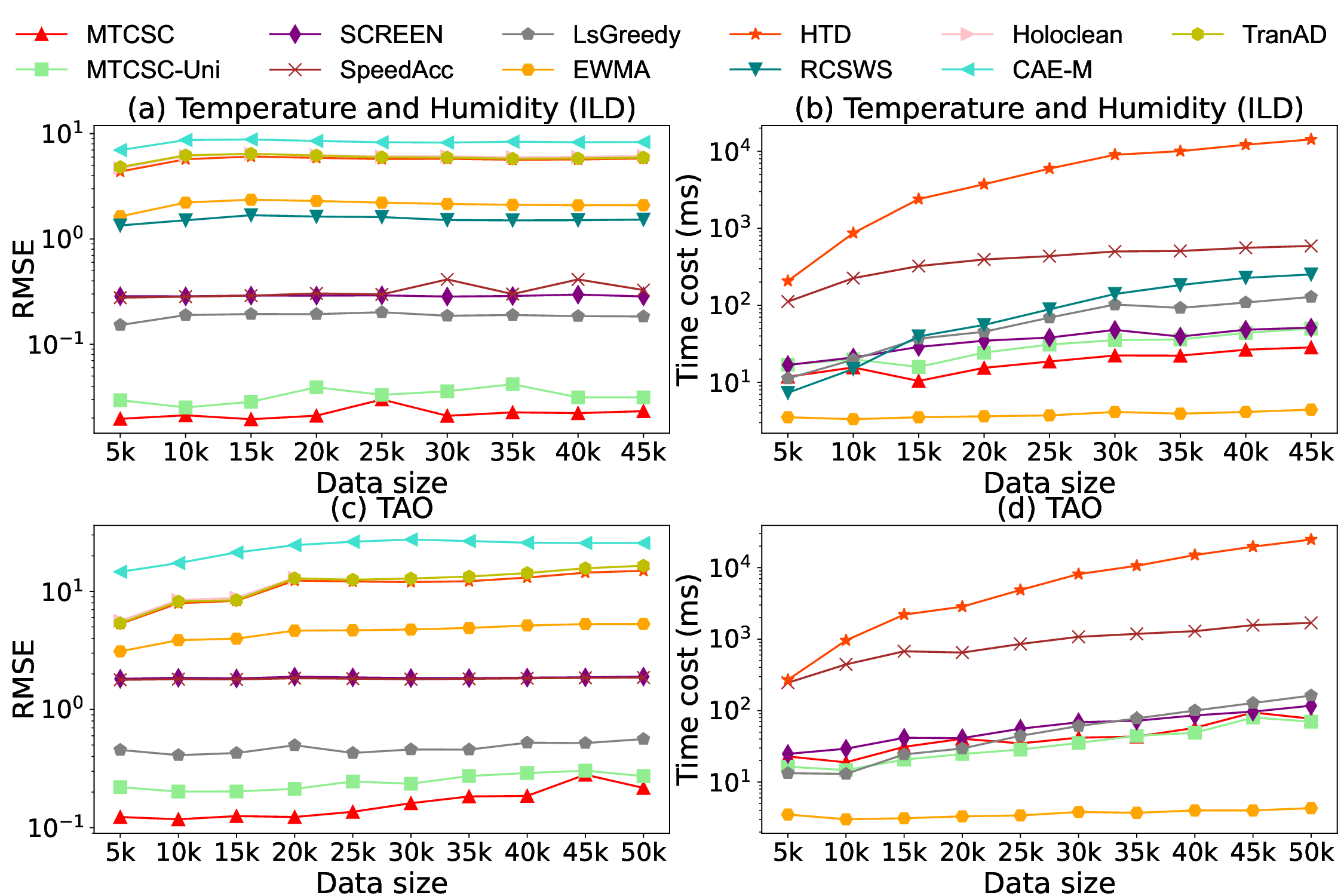}
  \caption{Varying data size on (a-b) ILD (c-d) TAO (together)}%
  \label{exp:size-together}
\end{figure}

Similarly, Figure \ref{exp:size-separate}(a) and Figure \ref{exp:size-together}(a) show that \textsf{MTCSC} and \textsf{MTCSC-Uni} outperform others on ILD over different data sizes.
In Figure \ref{exp:size-separate}(c), \textsf{LsGreedy} shows comparable performance to \textsf{MTCSC}, but is still worse than \textsf{MTCSC-Uni}.
In Figure \ref{exp:size-together}(c), the two proposed methods show slightly better performance than \textsf{LsGreedy}.
Figures \ref{exp:size-separate} and \ref{exp:size-together} confirm that \textsf{MTCSC} and \textsf{MTCSC-Uni} scale well with large datasets.

\subsubsection{Varying Error Pattern}
\label{sect:error-pattern}
%
We inject two error patterns in multivariate data: ``separate'' and ``together''.
(1) 
\begin{revised}
``Separate'' means that we inject errors into each of the dimensions independently.
For example, 2.5\% of the points will have errors in one dimension, while another 2.5\% will be affected in another dimension, when we inject 5\% errors into two-dimensional data.
\end{revised}
(2) 
\begin{revised}
``Together'' means that all dimensions of the randomly selected data points are affected simultaneously.
\end{revised}
This is to simulate the real case that a sensor can collect different types of data (e.g., temperature and humidity in ILD) at the same time, so that errors that occur due to physical or transmission failures should occur simultaneously at a given time.

Our \textsf{MTCSC} outperforms others (including \textsf{MTCSC-Uni}) under ``together'' pattern, as shown in Figures \ref{exp:drate-together} and \ref{exp:size-together}.
Under ``Separate'' pattern,
\textsf{MTCSC} performs similarly with the most competitive baseline \textsf{LsGreedy}, while \textsf{MTCSC-Uni} is a bit better.
In contrast, both \textsf{MTCSC} and \textsf{MTCSC-Uni} are better on ILD.
This is because as the number of dimensions increases, the error that occurs in a single dimension has little impact on the entire data, making it difficult for \textsf{MTCSC} to detect and repair.
%

\begin{figure}
\centering
\begin{minipage}{\expwidths}
\hspace{-0.5em}%
\includegraphics[width=\expwidths]{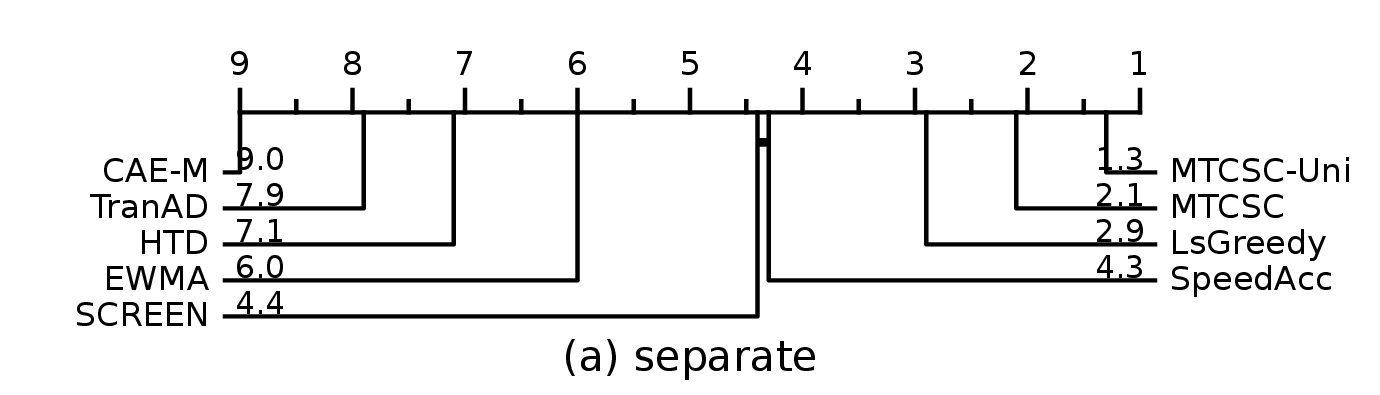}%
\end{minipage}
\hfill
\begin{minipage}{\expwidths}
\hspace{-0.5em}%
\includegraphics[width=\expwidths]{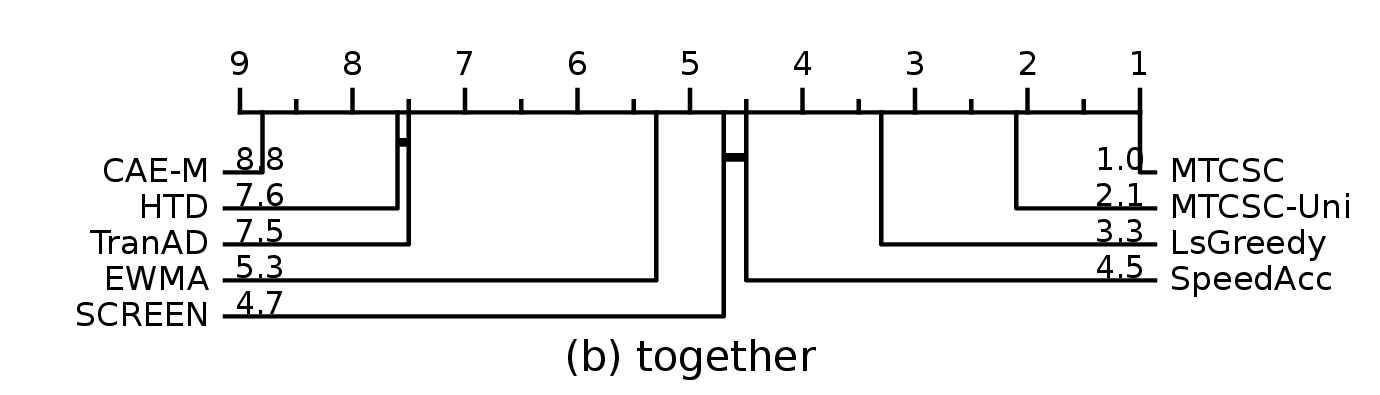}%
\end{minipage}
\caption{\begin{revised}Critical difference diagram on Error Pattern\end{revised}}
\label{exp:critical}
\end{figure}

\begin{revised}
Figure \ref{exp:critical} shows the critical difference diagram \cite{DBLP:journals/jmlr/Demsar06} for the error pattern.
Methods that are not connected by a bold line differ significantly in their average ranks.
In real cases where errors occur individually, it is therefore recommended to use \textsf{MTCSC-Uni}.
If errors can occur between multiple dimensions, as with GPS data, \textsf{MTCSC} is the best choice.
\end{revised}

\begin{figure}[t]
  \centering
  \includegraphics[width=\expwidths]{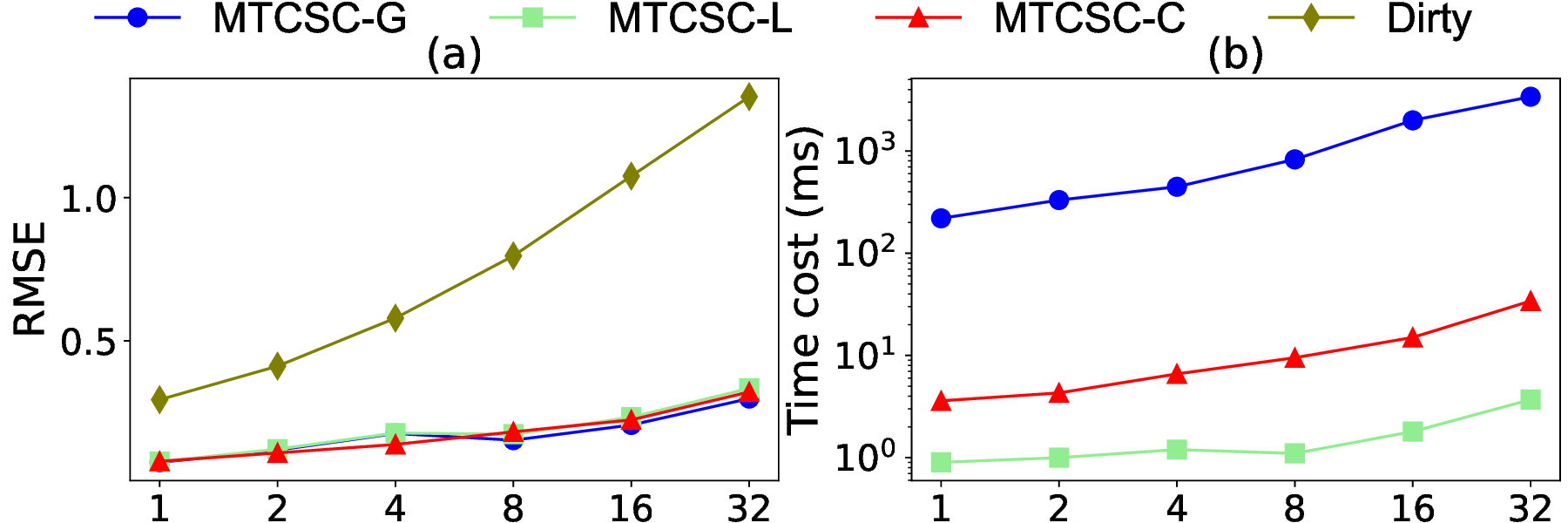}
  \caption{\begin{revised}Varying dimension on ECG\end{revised}}%
  \label{exp:dimension}
\end{figure}

\begin{revised}
\subsubsection{Varying Dimension $D$}
\label{sect:exp-dimension}
We conduct experiments with different dimensions on ECG (32 dimensions in total).
Figure \ref{exp:dimension}(b) verifies the theoretical analysis that the time complexity of \textsf{MTCSC-G}, \textsf{MTCSC-L} and \textsf{MTCSC-C} are linear functions of the number of dimensions and shows that they can be executed in a reasonable time.
Although it is difficult to determine the appropriate speed constraint as the dimension increases, our method can still achieve good results, as can be seen in Figure \ref{exp:dimension}(a).
\end{revised}

\subsection{Comparison on Real Errors}
\label{sect:gps}
We first compare the accuracy of error repair, and then validate the effectiveness of the adaptive method and explore its sensitivity.

\subsubsection{GPS Trajectory with Human Walking}
\label{sect:walking}

\begin{table}[t]
 \caption{GPS data with manually labeled ground truth}
 \label{table:gps}
 \centering
 \resizebox{0.8\expwidths}{!}{%
 \begin{tabular}{ccccc}
 \toprule
 & RMSE & repair distance & repair number \\
 \midrule
 \textsf{Dirty} & $1.3553$& - & - \\
 \textsf{MTCSC-G} & $0.4115$ & $0.1134$ & $163(1.52\%)$ \\
 \textsf{MTCSC-L} & $2.1569$ & $0.2241$ & $286(2.66\%)$ \\
 \textsf{MTCSC-C} & $\boldsymbol{0.3386}$ & $0.1265$ & $184(1.71\%)$ \\
 \textsf{MTCSC-Uni} & $0.4098$ & $0.1185$ & $160(1.49\%)$ \\
 \textsf{RCSWS} & $1.2096$ & $0.0571$ & $179(1.66\%)$ \\
 \textsf{SCREEN}& $0.9082$ & $0.0925$ & $284(2.64\%)$ \\
 \textsf{SpeedAcc} & $0.9065$ & $0.0928$ & $286(2.66\%)$ \\
 \textsf{LsGreedy} & $0.917$ & $0.061$ & $255(2.37\%)$ \\
 \textsf{EWMA} & $2.0859$ & $1.4236$ & $10753(99.99\%)$ \\
 \textsf{HTD} & $0.954$ & $0.0224$ & $41(0.38\%)$ \\
 \textsf{Holoclean} & $1.1733$ & - & - \\
 \textsf{CAE-M} & $159.5$ & - & - \\
 \textsf{TranAD} & $35.98$& - & - \\
 \bottomrule
 \end{tabular}
}
\end{table}

Table \ref{table:gps} shows the results on GPS dataset and the proposed \textsf{MTCSC-C} performs the best.
It is also worth noting that \textsf{MTCSC-L} is significantly worse than the others.
This is due to the characteristics of the collected walking data.
Walking through buildings or under bridges leads to consecutive errors, and these consecutive errors become the majority in our collection (the longest error sequence here contains 17 data points).
\textsf{MTCSC-L} only considers the previous repair and the next incoming point and is probably misled by the occurrence of consecutive errors.
In this scenario, we indeed observe the incapacity of almost all competitors.
\textsf{RCSWS} suffers from oversimplified considerations regarding the data, while \textsf{SCREEN} and \textsf{SpeedAcc} are hampered by the separate consideration of speed in the dimensions.
\textsf{LsGreedy} may not work optimally if only displacements of trajectory points occur without significant speed changes.
\textsf{HTD} cannot recognize most errors and remains unchanged.
It is possible that the learning-based methods cannot learn the pattern due to a lack of sufficient training data and achieve poor results.

\subsubsection{Different Correlations}
\label{sect:exp-correlation}
The superior performance of \textsf{MTCSC-C} compared to \textsf{MTCSC-Uni} in Table \ref{table:gps} shows that for data with correlations over Euclidean distance, the speed constraint over all dimensions together is the correct solution.
Interestingly, the results on ILD, TAO and ECG (Figures \ref{exp:drate-separate} to \ref{exp:size-together}) show that our \textsf{MTCSC} can still achieve good results as long as the speed is similarly scaled in different dimensions, even if there are no correlations between the dimensions over the Euclidean distance.
%

\subsubsection{Adaptive Speed with Different Transportation}
\label{sect:transportation}

\begin{figure}[t]
  \centering
  \includegraphics[width=\expwidths]{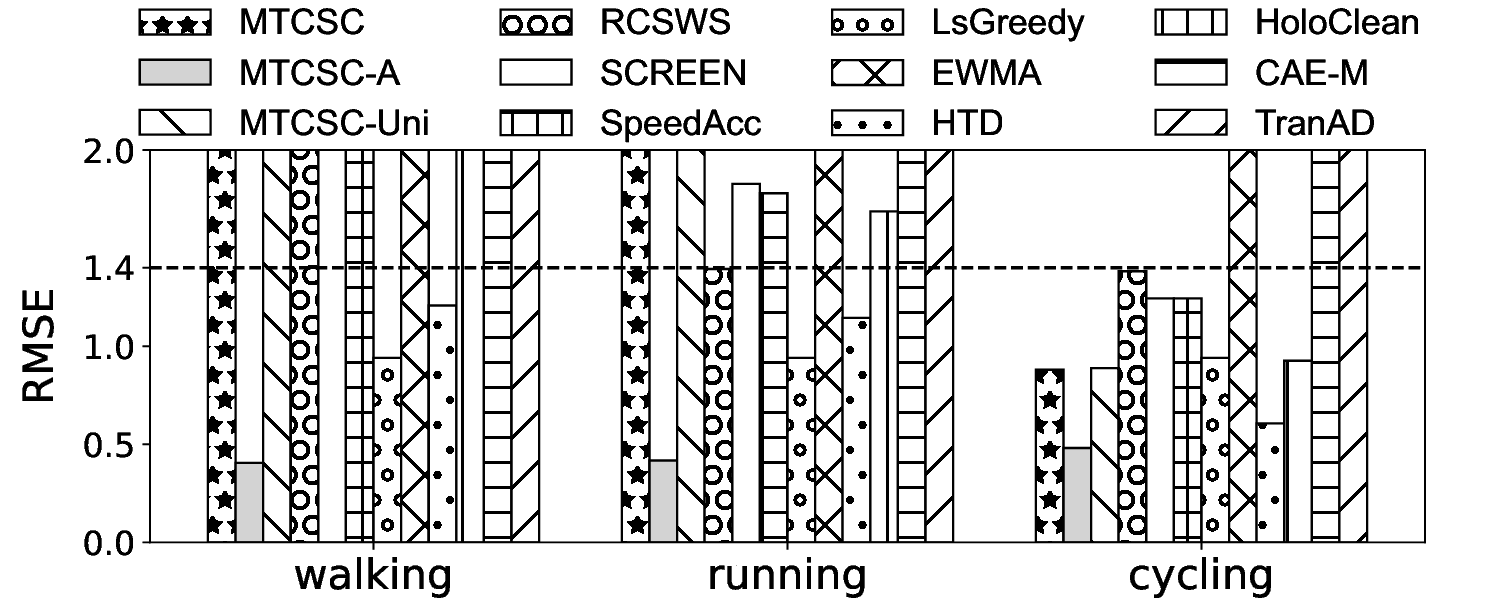}
  \caption{Adaptive Speed with Different Transportation}%
  \label{exp:adaptive-speed}
\end{figure}

\begin{figure}[t]
  \centering
  \includegraphics[width=\expwidths]{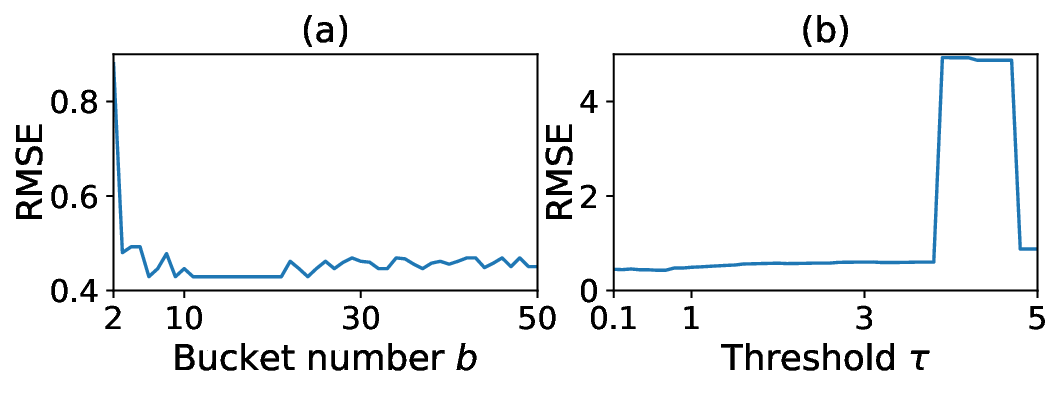}
  \caption{Sensitivity on bucket number $b$ and threshold $\tau$}%
  \label{exp:sensitivity}
\end{figure}

We conduct experiments on another collected GPS dataset with three transportation modes: walking, running and cycling.
We set the (initial) speed constraint corresponding to walking (1.6m/s), running (3.33m/s) and cycling (5.00m/s), respectively.
The hyper-parameters of \textsf{MTCSC-A} are $b = 6$, $\tau = 0.75$, $\mathit{w} = 150$ and $\beta = 0.75$.

The dashed line in Figure \ref{exp:adaptive-speed} represents the RMSE of the dirty data.
It is obvious that our \textsf{MTCSC-A} works optimally regardless of the initial setting.
When the initial speed is set to walking or running, similar to \textsf{EWMA}, the constraint-based methods behave poorly because they fix too many correct points (the cycling part).
\textsf{Lsgreedy} remains unaffected as it has nothing to do with the speed setting.
\textsf{HTD} performs relatively well with extra labels for the sake of execution (which is indeed unfair).

We perform experiments over bucket number $b$ and threshold $\tau$ (the other two are shown in Figure \ref{fig:KL}) to test the sensitivity.
Figure \ref{exp:sensitivity}(a) shows that \textsf{MTCSC-A} has a strong robustness to the bucket number $b$.
For the threshold $\tau$, the peak value observed in Figure \ref{exp:sensitivity}(b) means that a relatively large threshold delays the detection time of speed changes and thus causes an excessive change in correct points.
The subsequent sharp drop indicates that no more speed changes are detected if the threshold is even larger.
Therefore, the setting of the threshold should also be careful.

\subsection{Applications}
\label{sect:application}

\begin{figure}[t]
  \centering
  \includegraphics[width=\expwidths]{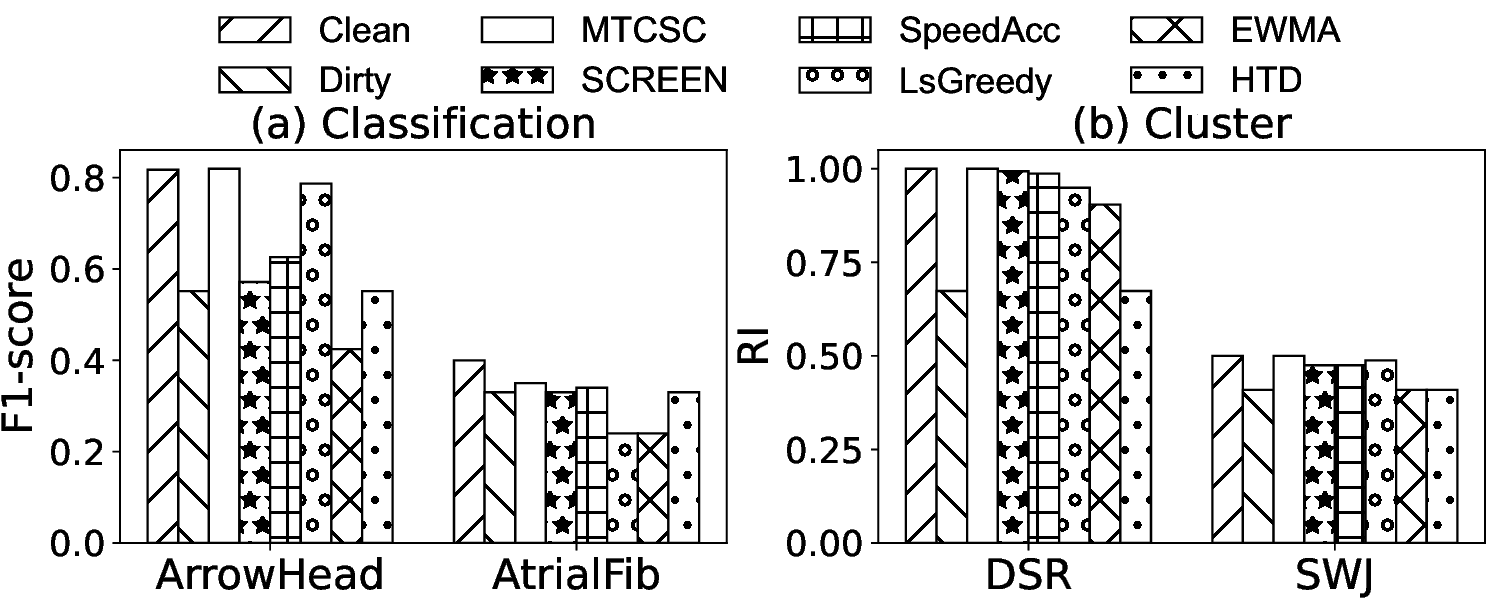}
  \caption{Classification and Cluster over clean, dirty and repaired data}%
  \label{exp:classification-cluster}
\end{figure}

We perform classification and clustering task to investigate the impact of different data cleaning methods on data mining applications.
Specifically, we evaluate the performance of the original data as ``Clean'', introduce 10\% random errors as ``Dirty'', and repair the data using different methods.
All four datasets are naturally divided into training and testing sets, where we introduce random errors only in the training set and leave the testing set unchanged.
\begin{revised}
For the multivariate time series AtrialFib and SWJ we inject errors with ``together'' pattern.
\end{revised}
The clustering task only uses the training sets.
For classification, we use \textsf{KNN} \cite{DBLP:journals/tit/CoverH67} classifier, select the best K via a grid search and take the F1 score \cite{aggarwal2016outlier} as the evaluation measure.
For clustering, we use \textsf{K-means} \cite{macqueen1967some} and take the RI \cite{rand1971objective} as the evaluation measure.

As expected, Figure \ref{exp:classification-cluster} shows that our \textsf{MTCSC} performs better than others and is similar to the results of ``Clean''.
After cleaning, the distance between the time series is calculated more correctly.

\subsection{Highlights and Limitations}
\label{sect:highlight}

We summarize the experimental highlights as follows:
1) Our proposal shows superior performance in cleaning time series compared to existing methods, with accuracy remaining stable even at high error rates (above 20\%) and good scalability for large datasets;
2) The \emph{minimum fix principle} can lead to smaller repair distances and increase repair accuracy compared to the minimum change principle as it suggests fewer changes to the originally correct points.
In addition, tracking the trends of succeeding points can increase accuracy compared to border repairs;
3) Under the scenario that the sensor suffers some accidents, that errors occur in each dimension simultaneously (together), the proposed \textsf{MTCSC} achieves a lower RMSE than other competitors, especially the univariate cleaning methods, which follows the intuition that we should consider the multivariate case as a whole;
\begin{revised}
4) Once we know that errors occur individually, \textsf{MTCSC-Uni} is a better choice;
\end{revised}
\begin{revised}
5) \textsf{MTCSC-C} can also perform an accurate cleaning of high-dimensional data in a reasonable time.
\end{revised}
6) \textsf{MTCSC-C} can be effective with real GPS data with consecutive errors.
7) We are surprised that \textsf{MTCSC} can achieve satisfactory repair results even when there is no correlation in terms of Euclidean distance between dimensions, once the speed in each dimension is scaled to a uniform magnitude (the smaller the difference, the better the result);
8) \textsf{MTCSC-A} accurately detects speed changes and thus achieves better repair results, relieving users of the difficult decision of setting an appropriate speed limit;
9) Data cleaning promotes the accuracy of data science applications, and ours behave best.
\begin{revised}
In short, we propose \textsf{MTCSC-C}, which uses only adjacent data points to identify the range of repair candidates, and \textsf{MTCSC-A} to dynamically determine the speed constraint to overcome the existing challenges 2 and 3 addressed in Section \ref{sect:solution}.
\end{revised}

On the other hand, our work still has some limitations.
1) As mentioned in \cite{DBLP:conf/sigmod/SongZWY15}, the arrival rate (the number of data points in a period) may vary and the data points may be delayed due to network reasons, but in this work we only assume that all points arrive in the correct order;
2) Although we introduce \textsf{MTCSC-A} to enable adaptive speed constraint, choosing an appropriate window size is still an open problem;
3) The reason for our highlight in terms of correlation between data dimensions is not yet clear 
\begin{revised}
(challenge 1).
\end{revised}


\section{Related Work}
\label{sect:related}

\subsection{Traditional Cleaning}

\paragraph{Constraint-based Cleaning} 
Constraint plays an important role in time series repairing. 
Holistic cleaning \cite{DBLP:conf/icde/ChuIP13} is the initial method supporting speed constraints, but it is designed for relational data. 
\cite{DBLP:journals/pvldb/GolabKKSS09} focuses on the difference between consecutive data points under sequential dependency, but it lacks precise expression of speed constraints. 
SCREEN\cite{DBLP:conf/sigmod/SongZWY15} proposes an online cleaning method under speed constraints but can only handle univariate data.
SpeedAcc \cite{DBLP:journals/tods/SongGZWY21} further considers the acceleration constraints in univariate time series. 
HTD \cite{DBLP:journals/ahswn/ZhouYZSMY22} considers both the dimensional correlation and temporal correlation, but relies heavily on the difference between labeled truth and the observations.
RCSWS \cite{DBLP:journals/tist/FangWYX22} repairs the GPS data based on range constraints and sliding window statistics.
Our \textsf{MTCSC} also takes advantage of data distribution and is applicable in multivariate time series cleaning.

\paragraph{Smoothing-based Cleaning} 
Moving average \cite{DBLP:books/daglib/0005327} is widely used to smooth time series data and make predictions. 
The simple moving average (SMA) calculates the unweighted mean of the last k data points and forecasts the next value.
Alternatively, the weighted moving average (WMA) assigns different weights to data based on their positions. 
The exponentially weighted moving average (EWMA) \cite{GARDNER2006637} assigns exponentially decreasing weights over time. 
These techniques always modify a large portion of the data, which may change the data distributions in original data.


\paragraph{Statistical-based Cleaning} 
Statistical-based methods are also employed in time series cleaning. 
LsGreedy \cite{DBLP:conf/sigmod/ZhangSW16} establishes a probability distribution model of speed changes between adjacent points, which overcomes the problem that speed constraint cannot detect small errors. 
\cite{DBLP:journals/envsoft/HillM10} presents an incremental clustering method where historical data is initially clustered, and the average values of the clusters are then used as the repaired result.
STPM \cite{DBLP:conf/icde/ZhengMC19} learns detailed data patterns from historical data and applies them to clean current data. 

\subsection{Machine Learning-based Cleaning}
HoloClean \cite{DBLP:journals/pvldb/RekatsinasCIR17} is a weakly supervised learning system that builds a probabilistic model for cleaning.
In order to cope with continuous data, we modify the quantization part and constraint parser part to promote the effectiveness on handling time series data with speed constraint. 
TranAD \cite{DBLP:journals/pvldb/TuliCJ22} is a prediction-based anomaly detection model that combines transformer-based encoder-decoder networks with adversarial training. 
CAE-M\cite{DBLP:journals/corr/abs-2107-12626} is an auto-encoder based anomaly detection model that combines convolution with a memory network that includes Autoregressive model and Bidirectional LSTM with Attention.

\section{Conclusion}
\label{sect:conclusion}

In this paper, we study the data cleaning problem in multivariate time series.
Existing methods based on the minimum change principle manage to minimize the repair distance, but lead to border repair.
On the other hand, current constraint-based methods all focus on univariate time series.
Even if the speed constraint is satisfied in each dimension, the constraint is still violated in multiple dimensions.
To solve the above problems, we propose \textsf{MTCSC}, which is based on the minimum fix principle and applies the speed constraint in all dimensions together.

We first formalize the cleaning problem and propose \textsf{MTCSC-G} to obtain the global repair.
Then, we propose \textsf{MTCSC-L} to support online cleaning.
To balance the efficiency and effectiveness, we develop \textsf{MTCSC-C} by considering both data distribution and speed constraint.
In addition, \textsf{MTCSC-A} is introduced to dynamically capture the speed constraint in real-world scenarios.
Experiments on various datasets with different error rates, data sizes and error patterns show that our proposals have good repair performance, robustness and scalability.
The reason why the speed constraint can be effective on datasets with weak/no correlations is interesting and remains an open problem.


\bibliographystyle{ACM-Reference-Format}
\bibliography{speed}

\newpage


\end{document}